%% file: superset.tex
\documentclass{article}

 \usepackage[preprint, nonatbib]{neurips_2019}

\usepackage[utf8]{inputenc} 
\usepackage[T1]{fontenc}    
\usepackage{hyperref}       
\usepackage{url}            
\usepackage{booktabs}       
\usepackage{amsfonts}       
\usepackage{nicefrac}       
\usepackage{microtype}      
\usepackage{enumitem}
\usepackage[cmex10]{amsmath} 
\usepackage{array}
\usepackage{threeparttable}
\usepackage{soul}
\usepackage{multirow}

\usepackage[compact]{titlesec}
\titlespacing{\section}{0pt}{*0}{*0}
\titlespacing{\subsection}{0pt}{*0}{*0}
\titlespacing{\subsubsection}{0pt}{*0}{*0}

\input{format.tex}

\theoremstyle{definition}
\newtheorem{property}{Property}
\crefname{Appendix}{Appendix}{Appendices}

\newcommand{\sign}{\textrm{sign}}
\crefname{algocf}{algorithm}{algorithms}
\crefname{definition}{definition}{definitions}
\Crefname{algocf}{Algorithm}{Algorithms}
\Crefname{definition}{Definition}{Definitions}

\title{Superset Technique for Approximate Recovery in One-Bit Compressed Sensing}

\author{%
  Larkin Flodin \\
  University of Massachusetts Amherst\\
  Amherst, MA 01003 \\
  \texttt{lflodin@cs.umass.edu} \\
 \And
 Venkata Gandikota \\
  University of Massachusetts Amherst\\
  Amherst, MA 01003 \\
 \texttt{gandikota.venkata@gmail.com} \\
 \And
Arya Mazumdar \\
  University of Massachusetts Amherst\\
  Amherst, MA 01003 \\
 \texttt{arya@cs.umass.edu} \\
}

\begin{document}

\maketitle

\begin{abstract}
One-bit compressed sensing (1bCS) is a method of signal acquisition under extreme measurement quantization that gives important insights on the limits of signal compression and analog-to-digital conversion. The setting is also equivalent to the problem of learning a sparse hyperplane-classifier. In this paper, we propose a novel approach for signal recovery in nonadaptive 1bCS that matches the sample complexity of the current best methods. We construct 1bCS matrices that are universal - i.e. work for all signals under a model - and at the same time recover very general random sparse signals with high probability. In our approach, we divide the set of samples (measurements) into two parts, and use the first part to recover the superset of the support of a sparse vector. The second set of measurements is then used to approximate the signal within the superset. While support recovery in 1bCS is well-studied, recovery of superset of the support requires fewer samples, and to our knowledge has not been previously considered for the purpose of approximate recovery of signals.
\end{abstract}

\section{Introduction}
\label{sec:intro}
Sparsity is a natural property of many real-world signals. For example, image and speech signals are sparse in the Fourier basis, which led to the  theory of compressed sensing, and more broadly, sampling theory \cite{landau1967sampling,DBLP:journals/tit/Donoho06}. 
In some important multivariate optimization problems with many optimal points, sparsity of the solution is also a measure of `simplicity' and insisting on sparsity is a common method of {\em regularization}~\cite{tibshirani1996regression}.
While recovering sparse vectors from linear measurements is a  well-studied topic, technological advances and increasing data size raises new questions. These include quantized and nonlinear signal acquisition models, such as 1-bit compressed sensing \cite{DBLP:conf/ciss/BoufounosB08}. In 1-bit compressed sensing, linear measurements of a sparse vector are quantized to only 1 bit, e.g. indicating whether the measurement outcome is positive or not, and the task is to recover the vector up to a prescribed Euclidean error with minimum number of measurements. Like compressed sensing, the overwhelming majority of the literature, including this paper, focuses on the nonadaptive setting for the problem.

One of the ways to approximately recover a sparse vector from 1-bit measurements is to use a subset of all the measurements to identify the support of the vector. Next, the remainder of the measurements can be used to approximate the vector within the support. Note that this second set of measurements is also predefined, and therefore the entire scheme is still nonadaptive. Such a method appears in the context of `universal' matrix designs in~\cite{GNJN13,ABK17}. The resulting schemes are the best known, in some sense, but still result in a large gap between the upper and lower bounds for approximate recovery of vectors.

In this paper we take steps to close these gaps by presenting a simple yet powerful idea. Instead of using a subset of the measurements to recover the support of the vector exactly, we propose using a (smaller) set of measurements to recover a superset of the support. The remainder of the measurements can then be used to better approximate the vector within the superset. We also present theoretical results providing a characterization of matrices that would yield universal schemes for all sparse vectors.

 
 \paragraph{Prior Results.}
 \label{sec:related_work}

While the compressed sensing framework was introduced in \cite{DBLP:journals/tit/Donoho06}, it was not until \cite{DBLP:conf/ciss/BoufounosB08} that 1-bit quantization of the measurements was considered as well, to try and combat the fact that taking real-valued measurements to arbitrary precision may not be practical in applications. Initially, the focus was primarily on approximately reconstructing the direction of the signal $\bfx$ (the quantization does not preserve any information about the magnitude of the signal, so all we can hope to reconstruct is the direction). However, in \cite{DBLP:conf/ciss/HauptB11} the problem of support recovery, as opposed to approximate vector reconstruction, was first considered and it was shown that $\calO{k \log n}$ measurements is sufficient to recover the support of a $k$-sparse signal in $\R^n$ with high probability. This was subsequently shown to be tight with the lower bound proven in \cite{AtiaS12}.

All the above results assume that a new measurement matrix is constructed for each sparse signal, and success is defined as either approximately recovering the signal up to error $\epsilon$ in the $\ell_2$ norm (for the approximate vector recovery problem), or exactly recovering the support of the signal (for the support recovery problem), with high probability. Generating a new matrix for each instance is not practical in all applications, which has led to interest in the ``universal'' versions of the above two problems, where a single matrix must work for support recovery or approximate recovery of all $k$-sparse real signals, with high probability.

Plan and Vershynin showed in \cite{DBLP:journals/tit/PlanV13} that both $\calO{\frac{k}{\epsilon^6} \log \frac{n}{k}}$ and $\calO{\frac{k}{\epsilon^5} \log^2 \frac{n}{k}}$ measurements suffice for universal approximate recovery. The dependence on $\epsilon$ was then improved significantly to $\widetilde{\mathcal{O}} \left(\frac{k}{\epsilon} \log \frac{n}{k} \right)$ by both \cite{GNJN13} and \cite{JLBB13}, using separate methods. Gopi et al. \cite{GNJN13} also considered the problem of universal support recovery, and showed that for that problem, $\calO{k^3 \log n}$ measurements is sufficient. They showed as well that if we restrict the entries of the signal to be nonnegative (which is the case for many real-world signals such as images), then $\calO{k^2 \log n}$ is sufficient for universal support recovery. The constructions of their measurement matrices are based primarily on combinatorial objects, specifically expanders and Union Free Families (UFFs).

Most recently, \cite{ABK17} showed that a modified version of the UFFs used in \cite{GNJN13} called ``Robust UFFs'' (RUFFs) can be used to improve the upper bound on universal support recovery to $\calO{k^2 \log n}$ for all real-valued signals, matching the previous upper bound for nonnegative signals, and showed this is nearly tight with a lower bound of $\Omega(k^2 \log n / \log k)$ for real signals.



\begin{remark}
\label{rem:prior_work_error}
In \cite{GNJN13}, a second algorithm (termed ``S-approx'') is provided for universal approximate recovery in addition to that mentioned above; this second algorithm first recovers the support, then performs approximate recovery within the support, and is stated in that work to have sample complexity $\widetilde{\mathcal{O}} \left(k^3 \log \frac{n}{k} + \frac{k}{\epsilon} \right)$. Similarly, the subsequent work of \cite{ABK17} improved the results for support recovery to take only $\calO{k^2 \log \frac{n}{k}}$ measurements, and thus claimed that this would yield an algorithm for approximate recovery with sample complexity $\widetilde{\mathcal{O}} \left(k^2 \log \frac{n}{k} + \frac{k}{\epsilon} \right)$.

We have confirmed in personal correspondence with the authors of \cite{ABK17} that this is an error; the second portion of approximate recovery within the support in fact requires $\widetilde{\mathcal{O}} \left(\frac{k}{\epsilon} \log \frac{n}{k} \right)$ samples. Furthermore, the authors of \cite{ABK17} told us they had correspondence with the authors of \cite{GNJN13}, who confirmed the same error is present in their algorithm. Thus, the correct sample complexity of this S-approx algorithm of \cite{GNJN13} is $\widetilde{\mathcal{O}}\left(k^3 \log \frac{n}{k} + \frac{k}{\epsilon} \log \frac{n}{k}\right)$, and the correct sample complexity of the method of \cite{ABK17} is $\widetilde{\mathcal{O}}\left(k^2 \log \frac{n}{k} + \frac{k}{\epsilon} \log \frac{n}{k}\right)$. Both are thus asymptotically inferior to the $\widetilde{\mathcal{O}} \left(\frac{k}{\epsilon} \log \frac{n}{k} \right)$ sample complexity of the first algorithm of \cite{GNJN13} and the algorithm of \cite{JLBB13}. The same error propagated also to earlier versions of this paper -- the ramifications of this are discussed in \cref{rem:our_error}.
\end{remark}
 
In tandem with the development of these theoretical results providing necessary and sufficient numbers of measurements for support recovery and approximate vector recovery, there has been a significant body of work in other directions on 1-bit compressed sensing, such as heuristic algorithms that perform well empirically, and tradeoffs between different parameters. More specifically, \cite{JLBB13} introduced a gradient-descent based algorithm called Binary Iterative Hard Thresholding (BIHT) which performs very well in practice; later, \cite{DBLP:conf/aistats/Li16} gave another heuristic algorithm which performs comparably well or better, and aims to allow for very efficient decoding after the measurements are taken. Other papers such as \cite{DBLP:conf/nips/SlawskiL15} have studied the tradeoff between the amount of quantization of the signal, and the necessary number of measurements.

\paragraph{Our Results.}
 We focus primarily on upper bounds in the universal setting, aiming to give constructions that work with high probability for all sparse vectors. In \cite{ABK17}, 3 major open questions are given regarding Universal 1-bit Compressed Sensing, which, paraphrasing, are as follows:
 \begin{enumerate}[topsep=0pt,noitemsep]
 \item How many measurements are necessary and sufficient for a matrix to be used to exactly recover all $k$-sparse binary vectors?
 \item What is the correct complexity (in terms of number of measurements) of universal $\epsilon$-approximate vector recovery for real signals?
 \item Can we obtain explicit (i.e. requiring time polynomial in $n$ and $k$) constructions of the Robust UFFs used for universal support recovery (yielding measurement matrices with $\calO{k^2 \log n}$ rows)?
 \end{enumerate}
 
In this work we investigate all three Open Questions. Our primary contribution is the ``superset technique'' which relies on ideas from the closely related sparse recovery problem of group testing \cite{du2000combinatorial}; in particular, we show in \cref{thm:eps_ub_generalization} that for a large class of signals including all nonnegative (and thus all binary) signals, our method matches existing upper bounds for approximate recovery by first recovering an $\calO{k}$-sized superset of the support rather than the exact support, then subsequently using Gaussian measurements. For binary signals this gives a sample complexity of $\calO{k^{3/2} \log n}$, and for nonnegative signals a sample complexity of $\widetilde{\mathcal{O}}\left(\frac{k}{\epsilon} \log \frac{n}{k})\right)$.

\begin{remark}
\label{rem:our_error}
In prior versions of this work, we also overlooked the error mentioned in \cref{rem:prior_work_error}. This led us to incorrectly claim that our method gave asymptotic sample complexity improvements for universal approximate recovery of non-negative and binary signals, when in fact the sample complexity of our method for these classes of signals turns out to be the same as other prior methods. In particular, the stated sample complexities of \cref{thm:eps_ub_generalization}, \cref{cor:nonnegative_apx_rec}, and \cref{cor:binary_exact_rec} were incorrect in these earlier versions. In this version we have corrected the errors in those results.
\end{remark}
 
Regarding Open Question 3, using results of Porat and Rothschild regarding weakly explicit constructions of Error-Correcting Codes (ECCs) on the Gilbert-Varshamov bound \cite{PR11}, we give a construction of Robust UFFs yielding measurement matrices for support recovery with $\calO{k^2 \log n}$ rows in time that is polynomial in $n$ (though not in $k$) in \cref{thm:explicit_uffs}. Based on a similar idea, we also give a weakly explicit construction for non-universal approximate recovery using only sightly more measurements than is optimal ($\calO{k \log^2 n}$ as opposed to $\calO{k \log \frac{n}{k}}$) in \cref{sec:apx_rec_whp}; to our knowledge, explicit constructions in the non-universal setting have not been studied previously. Furthermore, this result gives a single measurement matrix which works for almost all vectors, as opposed to typical non-universal results which work with high probability for a particular vector and matrix pair.

In \cref{sec:suff_cond_reals}, we give a sufficient condition generalizing the notion of RUFFs for a matrix to be used for universal recovery of a superset of the support for all real signals; while we do not provide constructions, this seems to be a promising direction for resolving Open Question 2.

\begin{table}
\tiny
  \caption{Upper and lower bounds for 1bCS problems with $k$-sparse signals}
  \label{prior_results_table}
  \centering
  \begin{tabular}{llcl}
    \toprule
    Problem & UB & Explicit UB & LB \\
    \midrule
    Universal Support Recovery ($\mathbf{x} \in \R^n$) & $\calO{k^2 \log n}$ \cite{ABK17}  & $\calO{k^2 \log n}^*$ & $\Omega(k^2 \log n / \log k)$ \cite{ABK17}     \\
    Universal $\calO{k}$-Superset Support Recovery ($\mathbf{x} \in \R^n_{\geq 0}$) & $\calO{k \log \frac{n}{k}}^*$ & $\calO{k^{1+o(1)} \log \frac{n}{k}}^*$ & $\Omega(k \log \frac{n}{k})$ \\ 
    Universal $\epsilon$-approximate Recovery ($\mathbf{x} \in \R^n$) & $\widetilde{\mathcal{O}}\left(\frac{k}{\epsilon} \log \frac{n}{k}\right)$ \cite{GNJN13}, \cite{JLBB13}  & -- & $\Omega(k \log \frac{n}{k} + \frac{k}{\epsilon})$ \cite{ABK17} \\
    Universal $\epsilon$-approximate Recovery ($\mathbf{x} \in \R^n_{\geq 0}$) & $\widetilde{\mathcal{O}}\left(\frac{k}{\epsilon} \log \frac{n}{k}\right)^\dagger$ \cite{GNJN13}, \cite{JLBB13} & -- & $\Omega(k \log \frac{n}{k})$     \\
    Universal Exact Recovery ($\mathbf{x} \in \set{0, 1}^n$) & $\widetilde{\mathcal{O}}\left(k^{3/2} \log \frac{n}{k}\right)^\dagger$ \cite{GNJN13}, \cite{JLBB13}  & -- & $\Omega(k \log \frac{n}{k})$     \\
    Non-Universal Support Recovery ($\mathbf{x} \in \R^n$) & $\calO{k \log n}$ \cite{AtiaS12}  & $\calO{k \log^2 n}^*$ & $\Omega(k \log \frac{n}{k})$ \cite{AtiaS12}    \\
    \bottomrule
  \end{tabular}\\
*Bound shown in this work. \hspace{5cm} $\dagger$ Bound matched in this work.
\end{table}

The best known upper and lower bounds for the various compressed sensing problems considered in this work are presented in \cref{prior_results_table}.

\section{Definitions}

We write $M_i$ for the $i$th row of the matrix $M$, and $M_{i,j}$ for the entry of $M$ in the $i$th row and $j$th column. We write vectors $\mathbf{x}$ in boldface, and write $\mathbf{x}_i$ for the $i$th component of the vector $\mathbf{x}$. The set $\set{1, 2, \dotsc, n}$ will be denoted by $[n]$, and for any set $S$ we write $\mathcal{P}(S)$ for the power set of $S$ (i.e. the set of all subsets of $S$).

We will write $\supp{\mathbf{x}} \subseteq [n]$ to mean the set of indices of nonzero components of $\mathbf{x}$ (so $\supp{\mathbf{x}} = \set{i : \mathbf{x}_i \neq 0}$), and $||\mathbf{x}||_0$ to denote $|\supp{\mathbf{x}}|$.

For a real number $y$, $\s{y}$ returns $1$ if $y$ is strictly positive, $-1$ if y is strictly negative, and $0$ if $y = 0$. While this technically returns more than one bit of information, if we had instead defined $\s{y}$ to be 1 when $y \geq 0$ and $-1$ otherwise, we could still determine whether $y=0$ by looking at $\s{y}, \s{-y}$, so this affects the numbers of measurements by only a constant factor. We will not concern ourselves with the constants involved in any of our results, so we have chosen to instead use the more convenient definition.

We will sometimes refer to constructions from the similar ``group testing'' problem in our results. To this end, we will use the symbol ``$\odot$'' to represent the group testing measurement between a measurement vector and a signal vector. Specifically, for a measurement $\bfm$ of length $n$ and signal $\bfx$ of length $n$, $\bfm \odot \bfx$ is equal to $1$ if $\supp{\bfm} \cap \supp{\bfx}$ is nonempty, and $0$ otherwise. We will also make use of the ``list-disjunct'' matrices used in some group testing constructions.

\begin{definition}
An $m \times n$ binary matrix $M$ is $(k, l)$-list disjunct if for any two disjoint sets $S, T \subseteq{\Col(M)}$ with $|S| = k, |T| = l$, there exists a row in $M$ in which some column from $T$ has a nonzero entry, but every column from $S$ has a zero.
\end{definition}

The primary use of such matrices is that in the group testing model, they can be used to recover a superset of size at most $k + l$ of the support of any $k$-sparse signal $\bfx$ from applying a simple decoding to the measurement results $M \odot \bfx$.

In the following definitions, we write $S$ for a generic set that is the domain of the signal. In this paper we consider signals with domain $\R, \R_{\geq 0}$ (nonnegative reals), and $\set{0, 1}$.

\begin{definition}
An $m \times n$ measurement matrix $M$ can be used for \textbf{Universal Support Recovery} of $k$-sparse $\mathbf{x} \in S^n$ (in $m$ measurements) if there exists a decoding function $f: \set{-1, 0, 1}^m \rightarrow \mathcal{P}([n])$ such that $f(\s{M \mathbf{x}}) = \supp{\mathbf{x}}$ for all $\mathbf{x}$ satisfying $||\mathbf{x}||_0 \leq k$.
\end{definition}

\begin{definition}
An $m \times n$ measurement matrix $M$ can be used for \textbf{Universal $\epsilon$-Approximate Recovery} of $k$-sparse $\mathbf{x} \in S^n$ (in $m$ measurements) if there exists a decoding function $f: \set{-1, 0, 1}^m \rightarrow S^n$ such that
\begin{equation*}
\left| \left |\frac{\mathbf{x}}{||\mathbf{x}||_2} - \frac{f(\s{M \mathbf{x}})}{||f(\s{M \mathbf{x}})||_2}\right| \right|_2 \leq \epsilon,
\end{equation*}
for all $\mathbf{x}$ with $||\mathbf{x}||_0 \leq k$.
\end{definition}

\section{Upper Bounds for Universal Approximate Recovery}
\label{sec:superset_ubs}

Here we present our main result, an upper bound on the number of measurements needed to perform universal $\epsilon$-approximate recovery for a large class of real vectors that includes all binary vectors and all nonnegative vectors. The general technique will be to first use what are known as ``list-disjunct'' matrices from the group testing literature to recover a superset of the support of the signal, then use Gaussian measurements to approximate the signal within the superset. Because the measurements in the second part are Gaussian, we can perform the recovery within the (initially unknown) superset nonadaptively.

First, we need a lemma stating the necessary and sufficient conditions on a signal vector $\bfx$ in order to be able to reconstruct the results of a single group testing measurement $\bfm \odot \bfx$ using sign measurements. To concisely state the condition, we introduce some notation: for a subset $S \subseteq [n]$ and vector $\bfx$ of length $n$, we write $\bfx|_S$ to mean the restriction of $\bfx$ to the indices of $S$.

\begin{lemma}
\label{lem:group_testing_nec_cond}
Let $\bfm \in \set{0, 1}^n$ and $\bfx \in \R^n$. Define $S = \supp{\bfm} \cap \supp{\bfx}$. If either $S$ is empty or $S$ is nonempty and $\bfm^T|_S \hspace{1mm} \bfx|_S  \neq 0$, we can reconstruct the result of the group testing measurement $\bfm \odot \bfx$ from the sign measurement $\s{\bfm^T \bfx}$.
\end{lemma}
\begin{proof}
We observe $\s{\bfm^T \bfx}$ and based on that must determine the value of $\bfm \odot \bfx$, or equivalently whether $S$ is empty or nonempty. If $\s{\bfm^T \bfx} \neq 0$ then $\bfm^T \bfx \neq 0$, so $S$ is nonempty and $\bfm \odot \bfx = 1$. Otherwise we have $\s{\bfm^T \bfx} = 0$, in which case we must have $\bfm^T \bfx = 0$. If $S$ were nonempty then we would have $\bfm^T|_S  \hspace{1mm} \bfx|_S = 0$, contradicting our assumption. Therefore in this case we must have $S$ empty and $\bfm \odot \bfx = 0$, so for $\bfx$ satisfying the above condition we can reconstruct the results of a group testing measurement.
\end{proof}

For convenience, we use the following property to mean that a signal $\bfx$ has the necessary property from \cref{lem:group_testing_nec_cond} with respect to every row of a matrix $M$.

\begin{property}
\label{apx_rec_necc_prop}
Let $M$ be an $m \times n$ matrix, and $\bfx$ a signal of length $n$. Define $S_i = \supp{M_i} \cap \supp{\bfx}$. Then for every row $M_i$ of $M$, either $S_i$ is empty, or $M_i^T|_{S_i} \hspace{1mm} \bfx|_{S_i} \neq 0$.
\end{property}

\begin{corollary}
\label{cor:eps_apx_superset}
Let $M$ be a $(k, l)$-list disjunct matrix, and $\bfx \in \R^n$ be a $k$-sparse real signal. If \cref{apx_rec_necc_prop} holds for $M$ and $\bfx$, then we can use the measurement matrix $M$ to recover a superset of size at most $k + l$ of the support of $\bfx$ using sign measurements.
\end{corollary}

Combining this corollary with results of \cite{de2005optimal}, there exist matrices with $\calO{k \log(\frac{n}{k})}$ rows which we can use to recover an $\calO{k}$-sized superset of the support of $\bfx$ using sign measurements, provided $\bfx$ satisfies the above condition. Strongly explicit constructions of these matrices exist also, although requiring $\calO{k^{1 + o(1)} \log n}$ rows \cite{Che13}.

The other result we need is one that tells us how many Gaussian measurements are necessary to approximately recover a real signal using maximum likelihood decoding. Similar results have been shown previously, such as in \cite{JLBB13}, but we provide a proof that we find to be more straightforward and adequate in the specific case relevant to our results in \cref{sec:eps_apx_pf}.

\begin{lemma}
\label{lem:eps_apx_real_vec}
There exists a measurement matrix $A$ for universal $\epsilon$-approximate recovery of $k$-sparse vectors in $\R^n$, provided that
\begin{equation*}
m = \Omega\left(\frac{k}{\epsilon}\log \frac{n^{3/2}}{k\epsilon}\right)
\end{equation*}
In particular, if we take $A$ to be an $m \times n$ matrix with all entries drawn i.i.d. from $\mathcal{N}(0,1)$ with $m$ specified as above, then with high probability $A$ is such a matrix.
\end{lemma}

Combining this with \cref{cor:eps_apx_superset} and the group testing constructions of \cite{de2005optimal}, we have the following theorem.

\begin{theorem}
\label{thm:eps_ub_generalization}
Let $M = \begin{bmatrix} M^{(1)} \\ M^{(2)} \end{bmatrix}$ where $M^{(1)}$ is a $(k, \calO{k})$-list disjunct matrix with $\calO{k \log \frac{n}{k}}$ rows, and $M^{(2)}$ is a matrix with $\calO{\frac{k}{\epsilon} \log \frac{n^{3/2}}{k \epsilon}}$ rows that can be used for $\epsilon$-approximate recovery within the superset, so $M$ consists of $\calO{k \log \frac{n}{k} + \frac{k}{\epsilon} \log \frac{n^{3/2}}{k \epsilon}}$ rows. Let $\bfx \in \R^n$ be a $k$-sparse signal. If \cref{apx_rec_necc_prop} holds for $M^{(1)}$ and $\bfx$, then $M$ can be used for $\epsilon$-approximate recovery of $\bfx$.
\end{theorem}

\begin{remark*}
We note that the class of signal vectors $\bfx$ which satisfy the condition in \cref{thm:eps_ub_generalization} is actually quite large, in the sense that there is a natural probability distribution over all sparse signals $\bfx$ for which vectors violating the condition occur with probability 0. The details are laid out in \cref{lem:GT_CS_conversion}.
\end{remark*}

As special cases, we have upper bounds for nonnegative and binary signals. For ease of comparison with the other results, we assume the binary signal is rescaled to have unit norm, so has all entries either 0 or equal to $1 / \sqrt{||\bfx||_0}$.

\begin{corollary}
\label{cor:nonnegative_apx_rec}
Let $M = \begin{bmatrix} M^{(1)} \\ M^{(2)} \end{bmatrix}$ where $M^{(1)}$ is a $(k, \calO{k})$-list disjunct matrix with $\calO{k \log \frac{n}{k}}$ rows, and $M^{(2)}$ is a matrix with $\calO{\frac{k}{\epsilon} \log \frac{n^{3/2}}{k\epsilon}}$ rows that can be used for $\epsilon$-approximate recovery within the superset as in \cref{lem:eps_apx_real_vec}, so $M$ consists of $\calO{k \log\frac{n}{k} + \frac{k}{\epsilon} \log \frac{n^{3/2}}{k \epsilon}}$ rows. Let $\bfx \in \R^n$ be a $k$-sparse signal. If all entries of $\bfx$ are nonnegative, then $M$ can be used for $\epsilon$-approximate recovery of $\bfx$.
\end{corollary}
\begin{proof}
In light of \cref{thm:eps_ub_generalization}, we need only note that as all entries of $M^{(1)}$ and $\bfx$ are nonnegative, \cref{apx_rec_necc_prop} is satisfied for $M^{(1)}$ and $\bfx$.
\end{proof}

\begin{corollary}
\label{cor:binary_exact_rec}
Let $M = \begin{bmatrix} M^{(1)} \\ M^{(2)} \end{bmatrix}$ where $M^{(1)}$ is a $(k, \calO{k})$-list disjunct matrix with $\calO{k \log \frac{n}{k}}$ rows, and $M^{(2)}$ is a matrix with $\calO{k^{3/2} \log \frac{n^{3/2}}{\sqrt{k}}}$ rows that can be used for $\epsilon$-approximate recovery (with $\epsilon < 1 / \sqrt{k}$) within the superset as in \cref{cor:eps_apx_superset}, so $M$ consists of $\calO{k \log \frac{n}{k} + k^{3/2} \log \frac{n^{3/2}}{\sqrt{k}}}$ rows. Let $\bfx \in \R^n$ be the $k$-sparse signal vector. If all nonzero entries of $\bfx$ are equal, then $M$ can be used for exact recovery of $\bfx$.
\end{corollary}

\begin{proof}
Here we use the fact that if we perform $\epsilon$-approximate recovery using $\epsilon < 1 / \sqrt{k}$ then as the minimum possible distance between any two $k$-sparse rescaled binary vectors is $1 / \sqrt{k}$, we will recover the signal vector exactly.
\end{proof}

\section{Explicit Constructions}
\label{sec:explicit_constructions}

\subsection{Explicit Robust UFFs from Error-Correcting Codes}
\label{sec:explicit_uff}

In this section we explain how to combine several existing results in order to explicitly construct Robust UFFs that can be used for support recovery of real vectors. This partially answers Open Problem 3 from \cite{ABK17}.

\begin{definition}
\label{def:robust_uff}
A family of sets $\mathcal{F} = \set{B_1, B_2, \ldots, B_n}$ with each $B_i \subseteq [m]$ is an $(n, m, d, k, \alpha)$-Robust-UFF if $|B_i| = d, \forall i$, and for every distinct $j_0, j_1, \ldots, j_k \in [n]$, $|B_{j_0} \cap (B_{j_1} \cup B_{j_2} \cup \cdots \cup B_{j_k})| < \alpha |B_{j_0}|$.
\end{definition}

It is shown in \cite{ABK17} that nonexplicit $(n, m, d, k, 1/2)$-Robust UFFs exist with $m = \calO{k^2 \log n}, d = \calO{k \log n}$ which can be used to exactly recover the support of any $k$-sparse real vector of length $n$ in $m$ measurements.

The results we will need are the following, where the $q$-ary entropy function $H_q$ is defined as
\begin{equation*}
H_q(x) = x \log_q (q - 1) - x \log_q x - (1 - x) \log_q (1 - x).
\end{equation*}

\begin{theorem}[\cite{PR11} Thm. 2]
\label{thm:porat_rothschild}
Let $q$ be a prime power, $m$ and $k$ positive integers, and $\delta \in [0, 1]$. Then if $k \leq (1 - H_q(\delta)) m$, we can construct a $q$-ary linear code with rate $\frac{k}{m}$ and relative distance $\delta$ in time $\calO{m q^k}$.
\end{theorem}

\begin{theorem}[\cite{ABK17} Prop. 17]
\label{thm:ecc_uff_conversion}
Given a $q$-ary error correcting code with rate $r$ and relative distance $(1 - \beta)$, we can construct a $(q^{rd}, qd, d, 1, \beta)$-Robust-UFF.
\end{theorem}

\begin{theorem}[\cite{ABK17} Prop. 15]
\label{thm:uff_equivalence}
If $\mathcal{F}$ is an $(n, m, d, 1, \alpha/k)$-Robust-UFF, then $\mathcal{F}$ is also an $(n, m, d, k,\alpha)$-Robust-UFF.
\end{theorem}
By combining the above three results, we have the following.
\begin{theorem}
\label{thm:explicit_uffs}
We can explicitly construct an $(n, m, d, k, \alpha)$-Robust UFF with $m = \calO{\frac{k^2 \log n}{\alpha^2}}$ and $d = \calO{\frac{k \log n}{\alpha}}$ in time $\calO{(k/\alpha)^k}$.
\end{theorem}
\begin{proof}
First, we instantiate \cref{thm:porat_rothschild} to obtain a $q$-ary code $\mathcal{C}$ of length $d$ with $q = \calO{k/\alpha}$, relative distance $\delta = \frac{k - \alpha}{k}$, and rate $r = 1 - H_q(\delta)$ in time $\calO{q^k}$.

Applying \cref{thm:ecc_uff_conversion} to this code results in an $(n, m, d, 1, \beta)$-Robust-UFF $\mathcal{F}$ where $n = q^{rd}$, $m = qd$, $\beta = 1 - \delta$. By \cref{thm:uff_equivalence}, $\mathcal{F}$ is also an $(n, m, d, k, \beta k)$-Robust UFF. Plugging back in the parameters of the original code,
\begin{eqnarray*}
m = qd = \frac{q \log n}{r \log q} = \frac{q \log n}{(1 - H_q((k - \alpha)/k)) \log q} = \calO{\frac{k^2 \log n}{\alpha^2}}, \\
\beta k = (1 - \delta) k = (1 - \frac{k - \alpha}{k}) k = k - (k - \alpha) = \alpha.
\end{eqnarray*}
\end{proof}

While the time needed for this construction is not polynomial in $k$ (and therefore the construction is not strongly explicit) as asked for in Open Question 3 of \cite{ABK17}, this at least demonstrates that there exist codes with sufficiently good parameters to yield Robust UFFs with $m = \calO{k^2 \log n}$.

\subsection{Non-Universal Approximate Recovery}
\label{sec:apx_rec_whp}

If instead of requiring our measurement matrices to be able to recover all $k$-sparse signals simultaneously (i.e. to be universal), we can instead require only that they are able to recover ``most'' $k$-sparse signals. Specifically, in this section we will assume that the sparse signal is generated in the following way: first a set of $k$ indices is chosen to be the support of the signal uniformly at random. Then, the signal is chosen to be a uniformly random vector from the unit sphere on those $k$ indices. We relax the requirement that the supports of all $k$-sparse signals can be recovered exactly (by some decoding) to the requirement that we can identify the support of a $k$-sparse signal with probability at least $1 - \delta$, where $\delta \in [0, 1)$. Note that even when $\delta = 0$, this is a weaker condition than universality, as the space of possible $k$-sparse signals is infinite.

It is shown in \cite{AtiaS12} that a random matrix construction using $\calO{k \log n}$ measurements suffices to recover the support with error probability approaching 0 as $k$ and $n$ approach infinity. The following theorem shows that we can explicitly construct a matrix which works in this setting, at the cost of slightly more measurements (about $\calO{k \log^2(n)}$).

\begin{theorem}
\label{thm:prob_explicit_UB}
We can explicitly construct measurement matrices for Support Recovery (of real vectors) with $m = \calO{k \frac{\log(n)}{\log k} \log(\frac{n}{\delta})}$ rows that can exactly determine the support of a $k$-sparse signal with probability at least $1 - \delta$, where the signals are generated by first choosing the size $k$ support uniformly at random, then choosing the signal to be a uniformly random vector on the sphere on those $k$ coordinates.
\end{theorem}

To prove this theorem, we need a lemma which explains how we can use sign measurements to ``simulate'' group testing measurements with high probability. Both the result and proof are similar to \cref{lem:group_testing_nec_cond}, with the main difference being that given the distribution described above, the vectors violating the necessary condition in \cref{lem:group_testing_nec_cond} occur with zero probability and so can be safely ignored. For this lemma, we do not need the further assumption made in \cref{thm:prob_explicit_UB} that the distribution over support sets is uniform. The proof is presented in \cref{sec:lem12_pf}.

\begin{lemma}
\label{lem:GT_CS_conversion}
Suppose we have a measurement vector $\mathbf{m} \in \set{0, 1}^n$, and a $k$-sparse signal $\mathbf{x} \in \R^n$. The signal $\mathbf{x}$ is generated randomly by first picking a subset of size $k$ from $[n]$ (using any distribution) to be the support, then taking $\bfx$ to be a uniformly random vector on the sphere on those $k$ coordinates. Then from $\sign(\bfm^T \bfx)$, we can determine the value of $\mathbf{m} \odot \mathbf{x}$ with probability 1.
\end{lemma}

As the above argument works with probability 1, we can easily extend it to an entire measurement matrix $M$ with any finite number of rows by a union bound, and recover all the group testing measurement results $M \odot \mathbf{x}$ with probability 1 as well. This means we can leverage the following result from \cite{Maz16}:

\begin{theorem}[\cite{Maz16} Thm. 5]
\label{thm:prob_GT}
When $\mathbf{x} \in \set{0, 1}^n$ is drawn uniformly at random among all $k$-sparse binary vectors, there exists an explicitly constructible group testing matrix $M$ with $m = \calO{\frac{k}{\log k} \log(n) \log(\frac{n}{\delta})}$ rows which can exactly identify $\mathbf{x}$ from observing the measurement results $M \odot \mathbf{x}$ with probability at least $1 - \delta$.
\end{theorem}

Combining this with the lemma above, we can use the matrix $M$ from \cref{thm:prob_GT} with $m = \calO{\frac{k}{\log k} \log n \log(\frac{n}{\delta})}$ rows (now representing sign measurements) to exactly determine the support of $\mathbf{x}$ with probability at least $1 - \delta$; we first use \cref{lem:GT_CS_conversion} to recover the results of the group testing tests $M \odot \mathbf{x}$ with probability 1, and can then apply the above theorem using the results of the group testing measurements.

We can also use this construction for approximate recovery rather than support recovery using \cref{lem:eps_apx_real_vec}, by appending $\calO{\frac{k}{\epsilon}\log \frac k\epsilon}$ rows of Gaussian measurements to $M$, first recovering the exact support, then doing approximate recovery within that support. This gives a matrix with about $\calO{k \log^2(n) + \frac{k}{\epsilon}\log \frac k\epsilon}$ rows for non-universal approximate recovery of real signals, where the top portion is explicit.

\begin{remark*}
Above, we have shown that in the non-universal setting, we can use constructions from group testing to recover the exact support with high probability, and then subsequently perform approximate recovery within that exact support. If we are interested only in performing approximate recovery, we can apply our superset technique here as well; \cref{lem:GT_CS_conversion} implies also that using a $(k, \calO{k})$-list disjunct matrix we can with probability 1 recover an $\calO{k}$-sized superset of the support, and such matrices exist with $\calO{k \log \frac{n}{k}}$ rows. Following this, we can use $\calO{\frac{k}{\epsilon} \log \frac{k}{\epsilon}}$ more Gaussian measurements to recover the signal within the superset. This gives a non-universal matrix with $\calO{k \log \frac{n}{k} + \frac{k}{\epsilon} \log \frac{k}{\epsilon}}$ rows for approximate recovery, the top part of which can be made strongly explicit with only slightly more measurements ($\calO{k^{1 + o(1)} \log \frac{n}{k}}$ vs. $\calO{k \log \frac{n}{k}}$).
\end{remark*}

\section{Experiments}
\label{sec:experiments}

In this section, we present some empirical results relating to the use of our superset technique in approximate vector recovery for real-valued signals. To do so, we compare the average error (in $\ell_2$ norm) of the reconstructed vector from using an ``all Gaussian'' measurement matrix to first using a small number of measurements to recover a superset of the support of the signal, then using the remainder of the measurements to recover the signal within that superset via Gaussian measurements. We have used the well-known BIHT algorithm of \cite{JLBB13} for recovery of the vector both using the all Gaussian matrix and within the superset, but we emphasize that this superset technique is highly general, and could just as easily be applied on top of other decoding algorithms that use only Gaussian measurements, such as the ``QCoSaMP'' algorithm of \cite{DBLP:conf/ita/ShiCGTN16}.

To generate random signals $\bfx$, we first choose a size $k$ support uniformly at random among the $\binom{n}{k}$ possibilities, then for each coordinate in the chosen support, generate a random value from $\mathcal{N}(0, 1)$. The vector is then rescaled so that $||\bfx||_2 = 1$.

\begin{figure}[H]
        \begin{subfigure}[b]{0.5\textwidth}
                \centering
                \includegraphics[width=.85\linewidth]{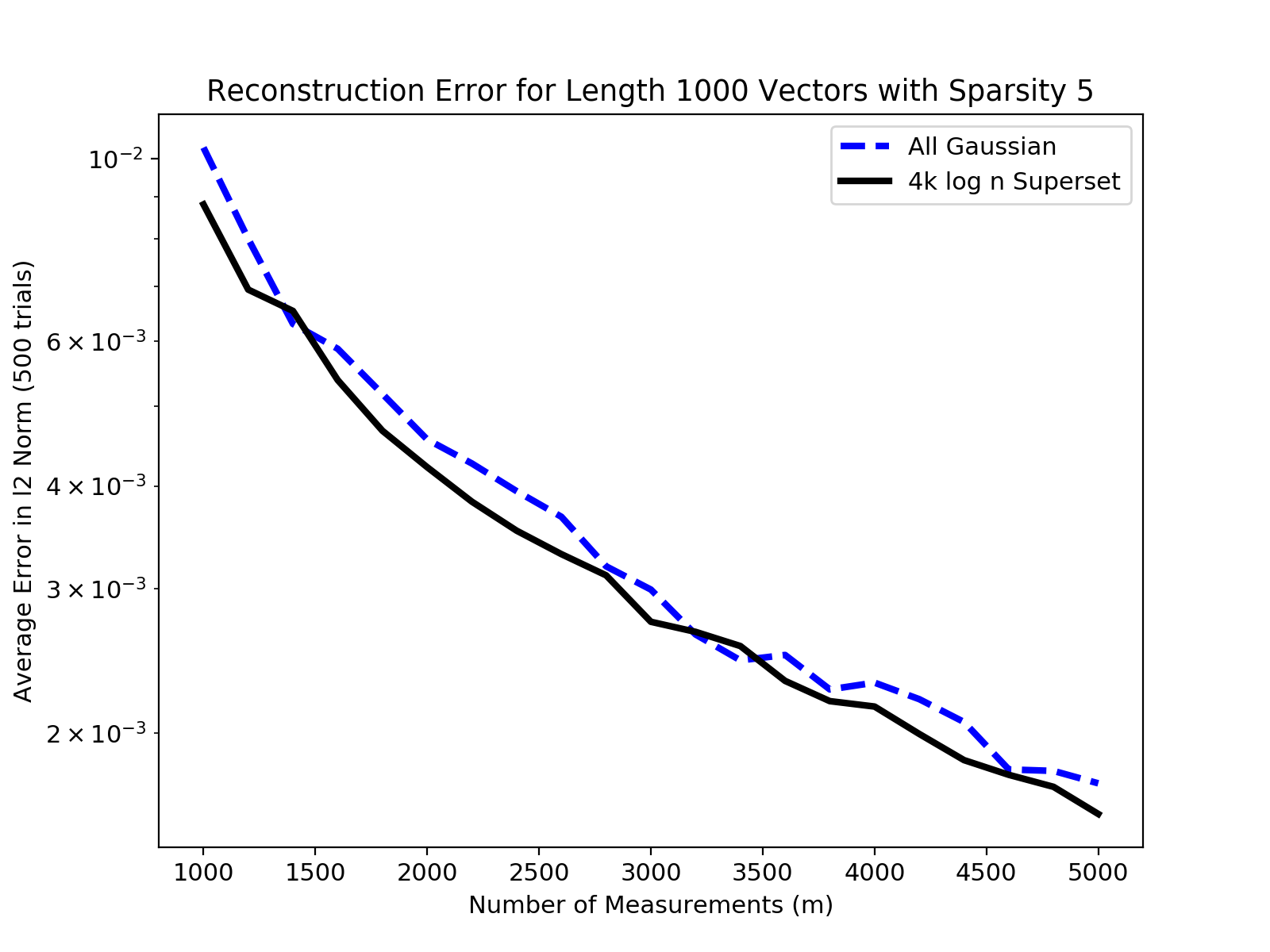}
                \caption{$n=1000, k=5$}
                \label{fig:k5}
        \end{subfigure}%
        \begin{subfigure}[b]{0.5\textwidth}
                \centering
                \includegraphics[width=.85\linewidth]{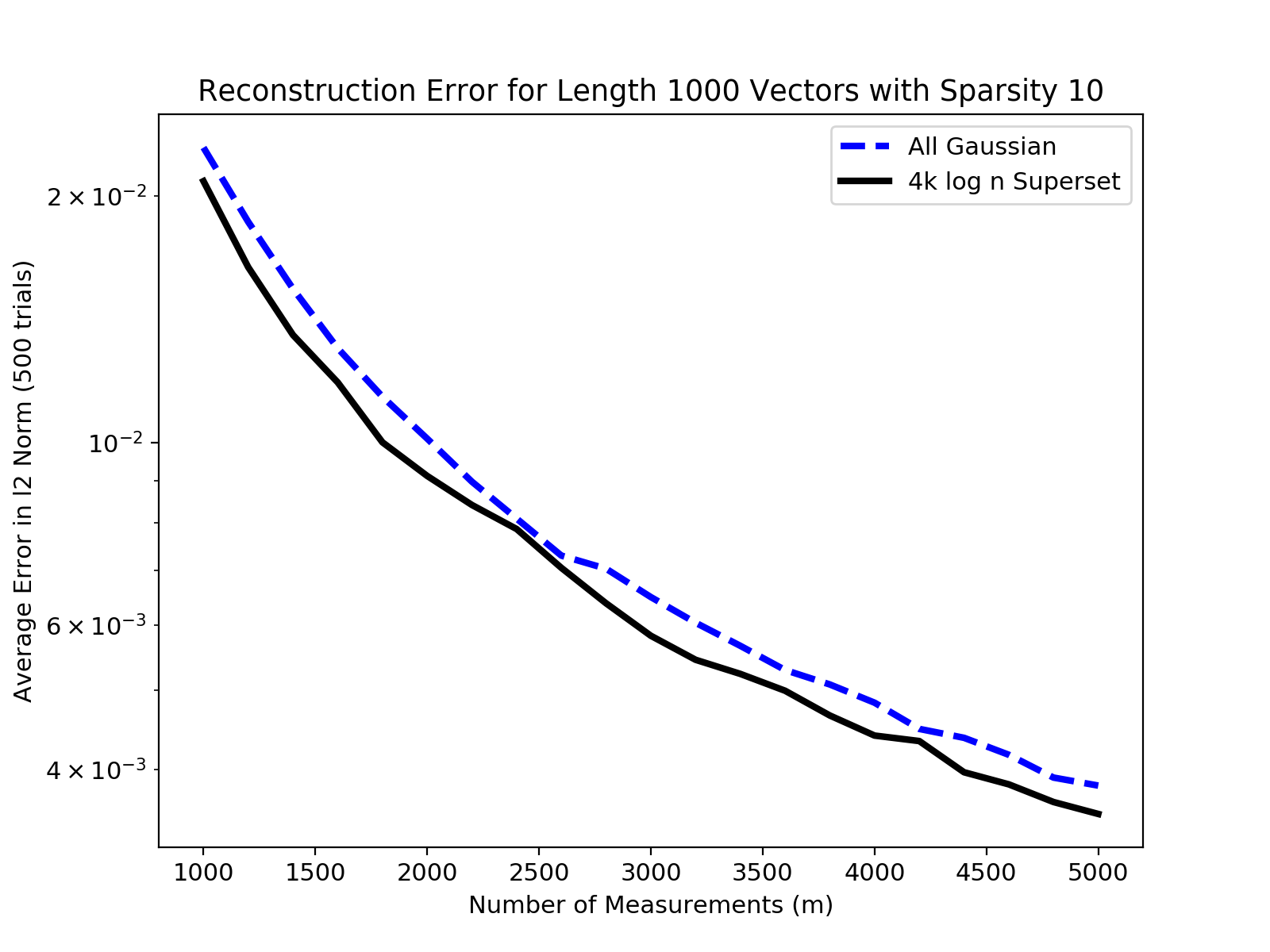}
                \caption{$n=1000, k=10$}
                \label{fig:k10}
        \end{subfigure}%
        \\
        \begin{subfigure}[b]{0.5\textwidth}
                \centering
                \includegraphics[width=.85\linewidth]{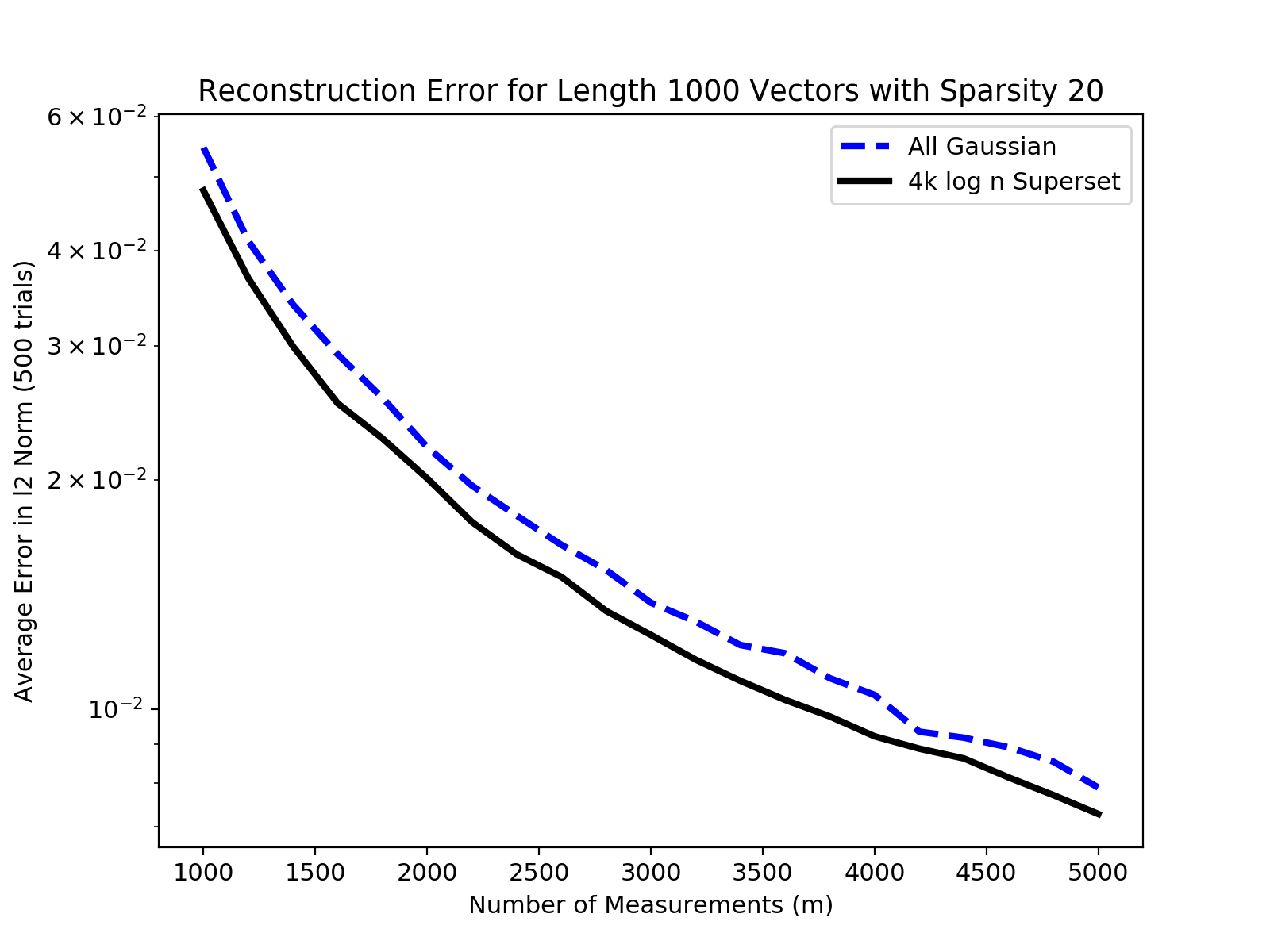}
                \caption{$n=1000, k=20$}
                \label{fig:k20}
        \end{subfigure}%
        \begin{subfigure}[b]{0.5\textwidth}
                \centering
                \includegraphics[width=.85\linewidth]{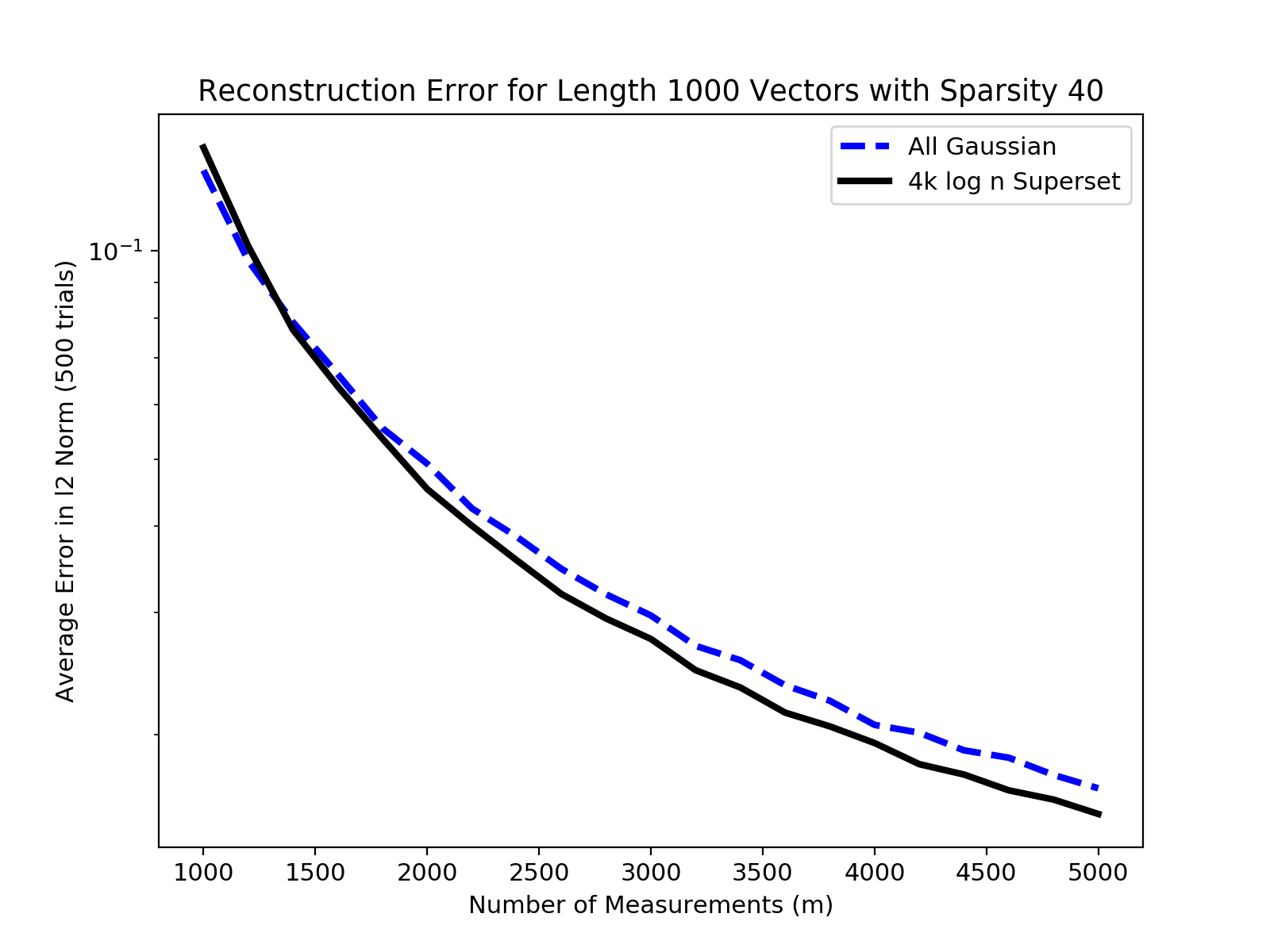}
                \caption{$n=1000, k=40$}
                \label{fig:k40}
        \end{subfigure}
        \caption{Average error of reconstruction for different sparsity levels with and without use of matrix for superset of support recovery}
        \label{fig:error_graphs}
\end{figure}

For the dotted lines in \cref{fig:error_graphs} labeled ``all Gaussian,'' for each value of $(n, m, k)$ we performed 500 trials in which we generated an $m \times n$ matrix with all entries in $\mathcal{N}(0, 1)$. We then used BIHT (run either until convergence or 1000 iterations, as there is no convergence guarantee) to recover the signal from the measurement matrix and measurement outcomes.

For the solid lines in \cref{fig:error_graphs} labeled ``$4k \log n$ Superset,'' we again performed 500 trials for each value of $(n, m, k)$ where in each trial we generated a measurement matrix $M = \begin{bmatrix} M^{(1)} \\ M^{(2)} \end{bmatrix}$ with $m$ rows in total. Each entry of $M^{(1)}$ is a Bernoulli random variable that takes value 1 with probability $\frac{1}{k + 1}$ and value 0 with probability $\frac{k}{k + 1}$; there is evidence from the group testing literature \cite{AtiaS12, DBLP:journals/tit/AldridgeBJ14} that this probability is near-optimal in some regimes, and it appears also to perform well in practice; see \cref{sec:empirical} for some empirical evidence. The entries of $M^{(2)}$ are drawn from $\mathcal{N}(0, 1)$. We use a standard group testing decoding (i.e., remove any coordinates that appear in a test with result 0) to determine a superset based on $\bfy_1 = \sign(M^{(1)} \bfx)$, then use BIHT (again run either until convergence or 1000 iterations) to reconstruct $\bfx$ within the superset using the measurement results $\bfy_2 = \sign(M^{(2)} \bfx)$. The number of rows in $M^{(1)}$ is taken to be $m_1 = 4k \log_{10}(n)$ based on the fact that with high probability $Ck \log n$ rows for some constant $C$ should be sufficient to recover an $\calO{k}$-sized superset, and the remainder $m_2 = (m - m_1)$ of the measurements are used in $M^{(2)}$.

We display data only for larger values of $m$, to ensure there are sufficiently many rows in both portions of the measurement matrix. From \cref{fig:error_graphs} one can see that in this regime, using a small number of measurements to first recover a superset of the support provides a modest improvement in reconstruction error compared to the alternative. In the higher-error regime when there are simply not enough measurements to obtain an accurate reconstruction, as can be seen in the left side of the graph in \cref{fig:k40}, the two methods perform about the same. In the empirical setting, our superset of support recovery technique can be viewed as a very flexible and low overhead method of extending other existing 1bCS algorithms which use only Gaussian measurements, which are quite common.

{\em Acknowledgements:} This research is supported in part by NSF CCF awards 1618512, 1642658, and 1642550 and the UMass Center for Data Science.

\bibliographystyle{plain}
\bibliography{superset}
\newpage
\appendix

\begin{center}
{\Large Supplementary Material: Superset Technique for Approximate Recovery in One-Bit Compressed Sensing}
\end{center}
\input{appendix}

\end{document}

%% file: format.tex
\usepackage{amsmath,amsthm,mathtools}
\usepackage{amssymb}
\usepackage{mathrsfs}
\usepackage{algorithmic}

\usepackage{multirow}
\usepackage{array}

\usepackage{amsfonts}
\usepackage{graphicx}
\usepackage[linesnumbered]{algorithm2e}

\usepackage{color, verbatim}
\usepackage[usenames,dvipsnames]{xcolor}

\usepackage{epsfig}
\usepackage{subcaption}
\usepackage{epstopdf}
\usepackage{xargs}
\usepackage{hhline}

\usepackage{float}

\usepackage[capitalize, noabbrev]{cleveref}

\newcommand\nc\newcommand
\nc{\bb}[1]{\mathbb{#1}}
\renewcommand{\cal}[1]{\mathcal{#1}}

\DeclarePairedDelimiter{\set}{\lbrace}{\rbrace}
\DeclarePairedDelimiter{\br}{\lparen}{\rparen}

\nc\bfa{{\bf{a}}}\nc\bfA{{\boldsymbol A}}\nc\cA{{\cal A}} \nc\fA[1]{A\br*{#1}} \nc\fa[1]{a\br*{#1}}  \nc\rmA{\mathrm{A}} \nc\rma{\mathrm{a}}
\nc\bfb{{\bf{b}}}\nc\bfB{{\boldsymbol B}}\nc\cB{{\cal B}} \nc\fB[1]{B\br*{#1}} \nc\fb[1]{b\br*{#1}}  \nc\rmB{\mathrm{B}} \nc\rmb{\mathrm{b}}
\nc\bfc{{\bf{c}}}\nc\bfC{{\boldsymbol C}}\nc\cC{{\cal C}} \nc\fC[1]{C\br*{#1}} \nc\fc[1]{c\br*{#1}}  \nc\rmC{\mathrm{C}} \nc\rmc{\mathrm{c}}
\nc\bfd{{\bf{d}}}\nc\bfD{{\boldsymbol D}}\nc\cD{{\cal D}} \nc\fD[1]{D\br*{#1}} \nc\fd[1]{d\br*{#1}}  \nc\rmD{\mathrm{D}} \nc\rmd{\mathrm{d}}
\nc\bfe{{\bf{e}}}\nc\bfE{{\boldsymbol E}}\nc\cE{{\cal E}} \nc\fE[1]{E\br*{#1}} \nc\fe[1]{e\br*{#1}}  \nc\rmE{\mathrm{E}} \nc\rme{\mathrm{e}}
\nc\bff{{\bf{f}}}\nc\bfF{{\boldsymbol F}}\nc\cF{{\cal F}} \nc\fF[1]{F\br*{#1}} \nc\ff[1]{f\br*{#1}}  \nc\rmF{\mathrm{F}} \nc\rmf{\mathrm{f}}
\nc\bfg{{\bf{g}}}\nc\bfG{{\boldsymbol G}}\nc\cG{{\cal G}} \nc\fG[1]{G\br*{#1}} \nc\fg[1]{g\br*{#1}}  \nc\rmG{\mathrm{G}} \nc\rmg{\mathrm{g}}
\nc\bfh{{\bf{h}}}\nc\bfH{{\boldsymbol H}}\nc\cH{{\cal H}} \nc\fH[1]{H\br*{#1}} \nc\fh[1]{h\br*{#1}}  \nc\rmH{\mathrm{H}} \nc\rmh{\mathrm{h}}
\nc\bfi{{\bf{i}}}\nc\bfI{{\boldsymbol I}}\nc\cI{{\cal I}} \nc\fI[1]{I\br*{#1}} \nc\rmI{\mathrm{I}} \nc\rmi{\mathrm{i}}
\nc\bfj{{\bf{j}}}\nc\bfJ{{\boldsymbol J}}\nc\cJ{{\cal J}} \nc\fJ[1]{J\br*{#1}} \nc\fj[1]{j\br*{#1}} \nc\rmJ{\mathrm{J}} \nc\rmj{\mathrm{j}}
\nc\bfk{{\bf{k}}}\nc\bfK{{\boldsymbol K}}\nc\cK{{\cal K}} \nc\fK[1]{K\br*{#1}} \nc\fk[1]{k\br*{#1}} \nc\rmK{\mathrm{K}} \nc\rmk{\mathrm{k}}
\nc\bfl{{\bf{l}}}\nc\bfL{{\boldsymbol L}}\nc\cL{{\cal L}} \nc\fL[1]{L\br*{#1}} \nc\fl[1]{l\br*{#1}} \nc\rmL{\mathrm{L}} \nc\rml{\mathrm{l}}
\nc\bfm{{\bf{m}}}\nc\bfM{{\boldsymbol M}}\nc\cM{{\cal M}} \nc\fM[1]{M\br*{#1}} \nc\fm[1]{m\br*{#1}} \nc\rmM{\mathrm{M}} \nc\rmm{\mathrm{m}}
\nc\bfn{{\bf{n}}}\nc\bfN{{\boldsymbol N}}\nc\cN{{\cal N}} \nc\fN[1]{N\br*{#1}} \nc\fn[1]{n\br*{#1}} \nc\rmN{\mathrm{N}} \nc\rmn{\mathrm{n}}
\nc\bfo{{\bf{o}}}\nc\bfO{{\boldsymbol O}}\nc\cO{{\cal O}} \nc\fO[1]{O\br*{#1}} \nc\fo[1]{o\br*{#1}} \nc\rmO{\mathrm{O}} \nc\rmo{\mathrm{o}}
\nc\bfp{{\bf{p}}}\nc\bfP{{\boldsymbol P}}\nc\cP{{\cal P}} \nc\fP[1]{P\br*{#1}} \nc\fp[1]{p\br*{#1}} \nc\rmP{\mathrm{P}} \nc\rmp{\mathrm{p}}
\nc\bfq{{\bf{q}}}\nc\bfQ{{\boldsymbol Q}}\nc\cQ{{\cal Q}} \nc\fQ[1]{Q\br*{#1}} \nc\fq[1]{q\br*{#1}} \nc\rmQ{\mathrm{Q}} \nc\rmq{\mathrm{q}}
\nc\bfr{{\bf{r}}}\nc\bfR{{\boldsymbol R}}\nc\cR{{\cal R}} \nc\fR[1]{R\br*{#1}} \nc\fr[1]{r\br*{#1}} \nc\rmR{\mathrm{R}} \nc\rmr{\mathrm{r}}
\nc\bfs{{\bf{s}}}\nc\bfS{{\boldsymbol S}}\nc\cS{{\cal S}} \nc\fS[1]{S\br*{#1}} \nc\fs[1]{s\br*{#1}} \nc\rmS{\mathrm{S}} \nc\rms{\mathrm{s}}
\nc\bft{{\bf{t}}}\nc\bfT{{\boldsymbol T}}\nc\cT{{\cal T}} \nc\fT[1]{T\br*{#1}} \nc\ft[1]{t\br*{#1}} \nc\rmT{\mathrm{T}} \nc\rmt{\mathrm{t}}
\nc\bfu{{\bf{u}}}\nc\bfU{{\boldsymbol U}}\nc\cU{{\cal U}} \nc\fU[1]{U\br*{#1}} \nc\fu[1]{u\br*{#1}} \nc\rmU{\mathrm{U}} \nc\rmu{\mathrm{u}}
\nc\bfv{{\bf{v}}}\nc\bfV{{\boldsymbol V}}\nc\cV{{\cal V}} \nc\fV[1]{V\br*{#1}} \nc\fv[1]{v\br*{#1}} \nc\rmV{\mathrm{V}} \nc\rmv{\mathrm{v}}
\nc\bfw{{\bf{w}}}\nc\bfW{{\boldsymbol W}}\nc\cW{{\cal W}} \nc\fW[1]{W\br*{#1}} \nc\fw[1]{w\br*{#1}} \nc\rmW{\mathrm{W}} \nc\rmw{\mathrm{w}}
\nc\bfx{{\bf{x}}}\nc\bfX{{\boldsymbol X}}\nc\cX{{\cal X}} \nc\fX[1]{X\br*{#1}} \nc\fx[1]{x\br*{#1}} \nc\rmX{\mathrm{X}} \nc\rmx{\mathrm{x}}
\nc\bfy{{\bf{y}}}\nc\bfY{{\boldsymbol Y}}\nc\cY{{\cal Y}} \nc\fY[1]{Y\br*{#1}} \nc\fy[1]{y\br*{#1}} \nc\rmY{\mathrm{Y}} \nc\rmy{\mathrm{y}}
\nc\bfz{{\bf{z}}}\nc\bfZ{{\boldsymbol Z}}\nc\cZ{{\cal Z}} \nc\fZ[1]{Z\br*{#1}} \nc\fz[1]{z\br*{#1}} \nc\rmZ{\mathrm{Z}} \nc\rmz{\mathrm{z}}

\DeclareMathOperator{\Supp}{supp}

\DeclareMathOperator{\Col}{col}
\DeclareMathOperator{\wt}{wt}

\nc\defeq{\coloneqq}

\newcommand{\supp}[1]{\Supp\br*{#1}}

\newcommand{\s}[1]{\mathrm{sign}(#1)}

\newcommand{\calO}[1]{{\cal O}\left(#1\right)}

\DeclarePairedDelimiterX\Set[1]\{\}{#1}

\newcommand\R{{\mathbb R}}

\newtheorem{theorem}{Theorem}
\newtheorem*{theorem*}{Theorem}
\newtheorem{definition}{Definition}
\crefname{definition}{defn.}{defns}
\newtheorem{lemma}[theorem]{Lemma}
\newtheorem*{lemma*}{Lemma}

\newtheorem{corollary}[theorem]{Corollary}
\newtheorem{corollary*}[theorem]{Corollary}

\newtheorem{remark}{\indent Remark}
\newtheorem*{remark*}{Remark}

\newtheorem{fact}[theorem]{Fact}

%% file: appendix.tex

\section{Proof of Lemma ~\ref{lem:eps_apx_real_vec}}
\label{sec:eps_apx_pf}

We start with a lemma that lower bounds the probability that a single random measurement fully separates two radius $\delta$ balls around points $\bfx$ and $\bfy$. We will make use of the following facts in the proof.

\begin{fact}
\label{fact:exp_lb}
For all $x \in \mathbb{R}$, $1 - x < e^{-x}$.
\end{fact}

\begin{fact}
\label{fact:arccos_lb}
For all $x \in [0, 1]$, $\cos^{-1}(x) \ge \sqrt{2(1-x)}$.
\end{fact}
\begin{proof}
Let $y \in [0,1]$. By Taylor series of $\cos(y)$, we know that 
$\cos(y) \ge 1 - \frac{y^2}{2}$. Therefore, setting $x = 1 - \frac{y^2}{2}$, the result follows. 
\end{proof}

\begin{lemma}
\label{lem:ball_sep_prob}
Let $\bfx$ and $\bfy$ be $k$-sparse unit vectors in $\R^n$ with $||\bfx - \bfy||_2 > \epsilon$, and take $\bfh \in \R^n$ to be a random vector with entries drawn i.i.d. from $\mathcal{N}(0,1)$. Write $B_\delta(\bfx)$ for the set $\set{\bfp \in \R^n : || \bfx - \bfp||_2 \leq \delta}$. Then
\begin{equation*}
\Pr[\forall \bfp \in B_\delta(\bfx), \forall \bfq \in B_\delta(\bfy), \sign(\bfh^T \bfp) \neq \sign(\bfh^T \bfq)] \geq \frac{\epsilon - 2\delta \sqrt{n}}{\pi}.
\end{equation*}
\end{lemma}
\begin{proof}
Define $E$ to be the event above, that $\bfh$ fully separates $B_\delta(\bfx)$ and $B_\delta(\bfy)$. Note that in order for $E$ to occur, three things are necessary, which we will refer to as $E_1, E_2$, and $E_3$, respectively:
\begin{enumerate}
\item $\bfh$ separates $\bfx$ and $\bfy$ (i.e. $\sign(\bfh^T \bfx) \neq \sign(\bfh^T \bfy)$).
\item $\bfx$ is at least $\delta$-far from the hyperplane with $\bfh$ as its perpendicular vector.
\item $\bfy$ is at least $\delta$-far from the hyperplane with $\bfh$ as its perpendicular vector.
\end{enumerate}

As $\bfh$ is chosen independent of $\bfx$ and $\bfy$, we know $\Pr[E_2] = \Pr[E_3]$. Thus
\begin{equation*}
\Pr[E] \geq \Pr[E_1 \cap E_2 \cap E_3] = 1 - \Pr[E_1^c \cup E_2^c \cup E_3^c] \geq 1 - \Pr[E_1^c] - 2\Pr[E_2^c].
\end{equation*}

The probability of $E_1^c$ is directly proportional to the angle between $\bfx$ and $\bfy$. Specifically, we have
\begin{eqnarray*}
\Pr[E_1^c] =& 1 - \frac{\cos^{-1}(\bfx^T \bfy)}{\pi} &\\
<& 1 - \frac{\cos^{-1}(1 - \epsilon^2/2)}{\pi} &\\
\leq& \exp(-\frac{\cos^{-1}(1 - \epsilon^2/2)}{\pi}) & \textrm{(from \cref{fact:exp_lb}})\\
\leq& \exp(-\frac{\epsilon}{\pi}) &\textrm{(from \cref{fact:arccos_lb}}).\\
\end{eqnarray*}
Then by again using \cref{fact:exp_lb},
\begin{equation*}
1 - \Pr[E_1^c] \geq \frac{\epsilon}{\pi}.
\end{equation*}

Next we turn to $E_2$.

If we set $Y = \frac{\langle \bfx, \bfh \rangle}{||\bfx||_2 ||\bfh||_2} = \langle \bfx, \frac{\bfh}{||\bfh||_2} \rangle$, $Y$ is simply the inner product of a uniformly random vector on a $n-1$ dimensional sphere with a fixed unit vector. It is known that in this case, $X = \frac{Y + 1}{2}$ is a random variable drawn from a Beta distribution with parameters $\left(\frac{n-1}{2}, \frac{n-1}{2}\right)$. Then we have

\begin{eqnarray*}
\Pr[E_2^c] =& \Pr[|Y| \leq \delta] \\
=& \Pr[-\delta \leq 2X - 1 \leq \delta] \\
=& \Pr[\frac{1- \delta}{2} \leq X \leq \frac{1+\delta}{2}]\\
\leq& \delta f\left(\frac{1}{2}\right),
\end{eqnarray*}
where $f$ is the PDF of the Beta$\left(\frac{n-1}{2}, \frac{n-1}{2}\right)$ distribution, and the last inequality follows because the width of the interval is $\delta$ and the mode of this distribution is $\frac{1}{2}$. Now we seek to upper bound $f\left(\frac{1}{2}\right)$.

By definition,
\begin{equation*}
f\left(\frac{1}{2}\right) = \frac{(1/2)^{(n-3)/2} (1-1/2)^{(n-3)/2}}{\left( \Gamma\left(\frac{n-1}{2}\right) \cdot \Gamma\left( \frac{n-1}{2}\right)\right) / \Gamma(n-1)} = \left(\frac{1}{2}\right)^{n-3} \frac{\Gamma(n-1)}{\left( \Gamma \left( \frac{n-1}{2}\right)\right)^2}.
\end{equation*}

Now we can apply the following form of Stirling's formula that holds for all $n$:
\begin{equation*}
\sqrt{2 \pi} n^{n+(1/2)} e^{-n} \leq \Gamma(n+1) \leq n^{n+(1/2)} e^{-n+1},
\end{equation*}
which yields
\begin{eqnarray*}
f\left(\frac{1}{2}\right) \leq& \left( \frac{1}{2} \right)^{n-3} \frac{(n-2)^{n-(3/2)} e^{-n+3}}{(2 \pi) \left( \left(n-3\right)/2 \right)^{n-2} e^{-n+3}} \\
=& \frac{(n-2)^{n-3/2}}{\pi (n-3)^{n-2}} \\
\leq& \frac{(n-2)^{n-3/2}}{\pi (n-2)^{n-2}} \\
=& \frac{\sqrt{n}}{\pi},
\end{eqnarray*}
and thus
\begin{equation*}
\Pr[E_2^c] \leq \frac{\delta \sqrt{n}}{\pi}.
\end{equation*}

Combining together, we have
\begin{equation*}
\Pr[E] \geq 1 - \Pr[E_1^c] - 2\Pr[E_2^c] \geq \frac{\epsilon}{\pi} - \frac{2 \delta \sqrt{n}}{\pi},
\end{equation*}
as desired.

\end{proof}

Now suppose we have a cover of all $k$-sparse points on an $n$-dimensional sphere by a $\delta$-net such that each point is within distance $\delta$ of a net point. Then if any two sparse points on the sphere are at distance at least $\epsilon$ from each other, the closest net points to each one are at distance at least $\epsilon - 2 \delta$ from each other. If we can guarantee that the $\delta$-balls around every pair of points in our cover at distance at least $\epsilon' = \epsilon - 2 \delta$ from each other are separated by some measurement, then we will know that any unit vectors $\bfx$ and $\bfy$ at distance at least $\epsilon$ will have different measurement results; this is a necessary and sufficient condition for $\epsilon$-approximate recovery.

\begin{lemma*}[\cref{lem:eps_apx_real_vec}]
Taking $A \in \R^{m \times n}$ to have all entries drawn i.i.d. from $\mathcal{N}(0,1)$ yields a measurement matrix suitable for universal $\epsilon$-approximate recovery of $k$-sparse unit vectors in $\R^n$ with high probability, provided that
\begin{equation*}
m = \Omega\Big(\frac{k}{\epsilon}\log \frac{n^{3/2}}{k\epsilon}\Big).
\end{equation*}
\end{lemma*}
\begin{proof}

In order to $\delta$-cover all vectors with a fixed size $k$ support $\left(\frac{3}{\delta}\right)^k$ points suffices, and there are $\binom{n}{k}$ such supports, so in total our net will require
\begin{equation*}
S = \binom{n}{k} \left( \frac{3}{\delta}\right)^k
\end{equation*}
points.

Call the lower bound on the probability of two $\delta$-balls of points at distance at least $\epsilon$ being fully separated (computed in \cref{lem:ball_sep_prob}) $p_{\epsilon,\delta,k}$. For two particular net points the probability of the bad event that their $\delta$-balls are not fully separated by a particular measurement is at most $(1 - p_{\epsilon',\delta,k})$. Since there are $m$ such measurements generated independently, the probability no measurement fully separates the $\delta$-balls is at most $(1 - p_{\epsilon',\delta,k})^m$.

We can then union bound over all pairs of net points, and the probability that the $\delta$-balls around any pair of net points are not fully separated is at most
\begin{equation}
\label{eq:UB_total_error_prob}
S^2 (1 - p_{\epsilon',\delta,k})^m \leq \left(\frac{n e}{k}\right)^{2k} \left(\frac{3}{\delta}\right)^{2k} (1 - p_{\epsilon',\delta,k})^m.
\end{equation}

If we take $\delta = \frac{\epsilon}{3(\sqrt{n}+1)}$, then by \cref{lem:ball_sep_prob} we have
\begin{eqnarray*}
p_{\epsilon',\delta,k} \geq& \frac{\epsilon - 2 \delta (1 + \sqrt{n})}{\pi} \\
\geq& \frac{\epsilon - 2(\epsilon/(3(1+\sqrt{n}))(1+\sqrt{n})}{\pi}\\
=& \frac{\epsilon}{3\pi} \\
\geq& \frac{\epsilon}{10}.
\end{eqnarray*}
Substituting this back into \cref{eq:UB_total_error_prob}, our error probability is at most
\begin{equation*}
\left(\frac{n e}{k}\right)^{2k} \left(\frac{3}{\delta}\right)^{2k} \left(1 - \frac{\epsilon}{10}\right)^m.
\end{equation*}

Now if we choose $\eta$ to be our maximum allowable error probability and substitute $\delta = \frac{\epsilon}{3(1+\sqrt{n})}$, we have
\begin{equation*}
\eta \geq \left(\frac{n e}{k}\right)^{2k} \left(\frac{9 (\sqrt{n} + 1)}{\epsilon}\right)^{2k} \left(1 - \frac{\epsilon}{10}\right)^m
\end{equation*}
Taking the log of both sides and using the inequality $1 - \frac{\epsilon}{10} \leq \exp(-\epsilon/10)$ (\cref{fact:exp_lb}), we have
\begin{equation*}
\log \eta \geq 2k \log \frac{ne}{k} + 2k \log \frac{9 (\sqrt{n} + 1)}{\epsilon} - \frac{m \epsilon}{10},
\end{equation*}
and therefore it suffices to have, 
\begin{eqnarray*}
m \geq& \frac{10}{\epsilon} \left( 2k \log \frac{ne}{k} + 2k \log \frac{9 (\sqrt{n} + 1)}{\epsilon} + \log \frac{1}{\eta}\right) \\
=& \frac{10}{\epsilon} \left( 2k \log \frac{9en^{3/2} + 9en}{k\epsilon} + \log \frac{1}{\eta}\right).
\end{eqnarray*}

\end{proof}

\section{Proof of Lemma ~\ref{lem:GT_CS_conversion}}
\label{sec:lem12_pf}

\begin{lemma*}[\cref{lem:GT_CS_conversion}]
Suppose we have a known measurement vector $\mathbf{m} \in \set{0, 1}^n$, and an unknown $k$-sparse signal $\mathbf{x} \in \R^n$. The signal $\mathbf{x}$ is generated randomly by first picking a subset of size $k$ from $[n]$ (using any distribution) to be the support, then taking $\bfx$ to be a uniformly random vector on the sphere on those $k$ coordinates. Then from $\sign(\bfm^T \bfx)$, we can determine the value of $\mathbf{m} \odot \mathbf{x}$ with probability 1.
\end{lemma*}

\begin{proof}
We assume without loss of generality that $\bfx$ is supported on the first $k$ coordinates; the remainder of the argument does not depend specifically on the choice of support, so this is purely for notational convenience. If $\sign(\mathbf{m}^T  \mathbf{x}) \neq 0$, then immediately we must have $\mathbf{m} \odot \mathbf{x} = 1$, as $\mathbf{m}^T \mathbf{x} \neq 0$.

Otherwise if $\sign(\mathbf{m}^T \mathbf{x}) = 0$, we must have $\mathbf{m}^T \mathbf{x} = 0$. This leaves two cases: either $\mathbf{m} \odot \mathbf{x} = 0$, or $\mathbf{x}$ is orthogonal to $\mathbf{m}$ and $\mathbf{m} \odot \mathbf{x} = 1$. In the latter case $\mathbf{x}$ satisfies the equation
\begin{equation*}
\sum_{i=1}^k \mathbf{m}_i \mathbf{x}_i = 0 \iff \mathbf{m}_1 \mathbf{x}_1 = -\left( \sum_{i=2}^k \mathbf{m}_i \mathbf{x}_i \right).
\end{equation*}

Let $\bfz$ be a random vector formed by using the same distribution as that used to determine the support of $\bfx$ in order to determine the support, then within that support drawing $k$ variables $Z_i \sim \mathcal{N}(0, 1)$ to be the $k$ coordinates, and finally rescaling so that $||\bfz||_2 = 1$. It is well-known that the distribution of such $\bfz$ is identical to the distribution of $\bfx$, thus the probability that $\bfz$ is orthogonal to $\bfm$ is the same as the probability that $\bfx$ is orthogonal to $\bfm$. We proceed by showing the probability $\bfz$ is orthogonal to $\bfm$ is 0.

If $\bfz$ is orthogonal to $\bfm$, then as above we must have
\begin{align*}
&\frac{\mathbf{m}_1 Z_1}{||\mathbf{z}||_2} = -\frac{\left( \sum_{i=2}^k \mathbf{m}_i Z_i \right)}{||\mathbf{z}||_2} \\
\implies& Z_1 = -\left( \sum_{i=2}^k \mathbf{m}_i Z_i \right) / \mathbf{m}_1.
\end{align*}
Thus in order for $\mathbf{z}$ to lie in the nullspace of $\mathbf{m}$, it is necessary that $Z_1$ takes a specific value determined by the other $k-1$ $Z_i$; as $Z_1$ is drawn independently of the other $Z_i$ and from a continuous distribution, this happens with probability 0. We conclude that the same is true for $\mathbf{x}$, and thus when $\sign(\mathbf{m}^T \mathbf{x}) = 0$ we assume that $\mathbf{m} \odot \mathbf{x} = 0$, and are correct with probability 1.
\end{proof}


\section{Empirical Evidence for Experimental Choice of Bernoulli Probability}
\label{sec:empirical}

In this section, we provide some empirical evidence that the choice of $\frac{1}{k+1}$ for the Bernoulli probability of the experiments in \cref{sec:experiments} is reasonable.

\begin{figure}
        \begin{subfigure}[H]{0.33\textwidth}
                \centering
                \includegraphics[width=.85\linewidth]{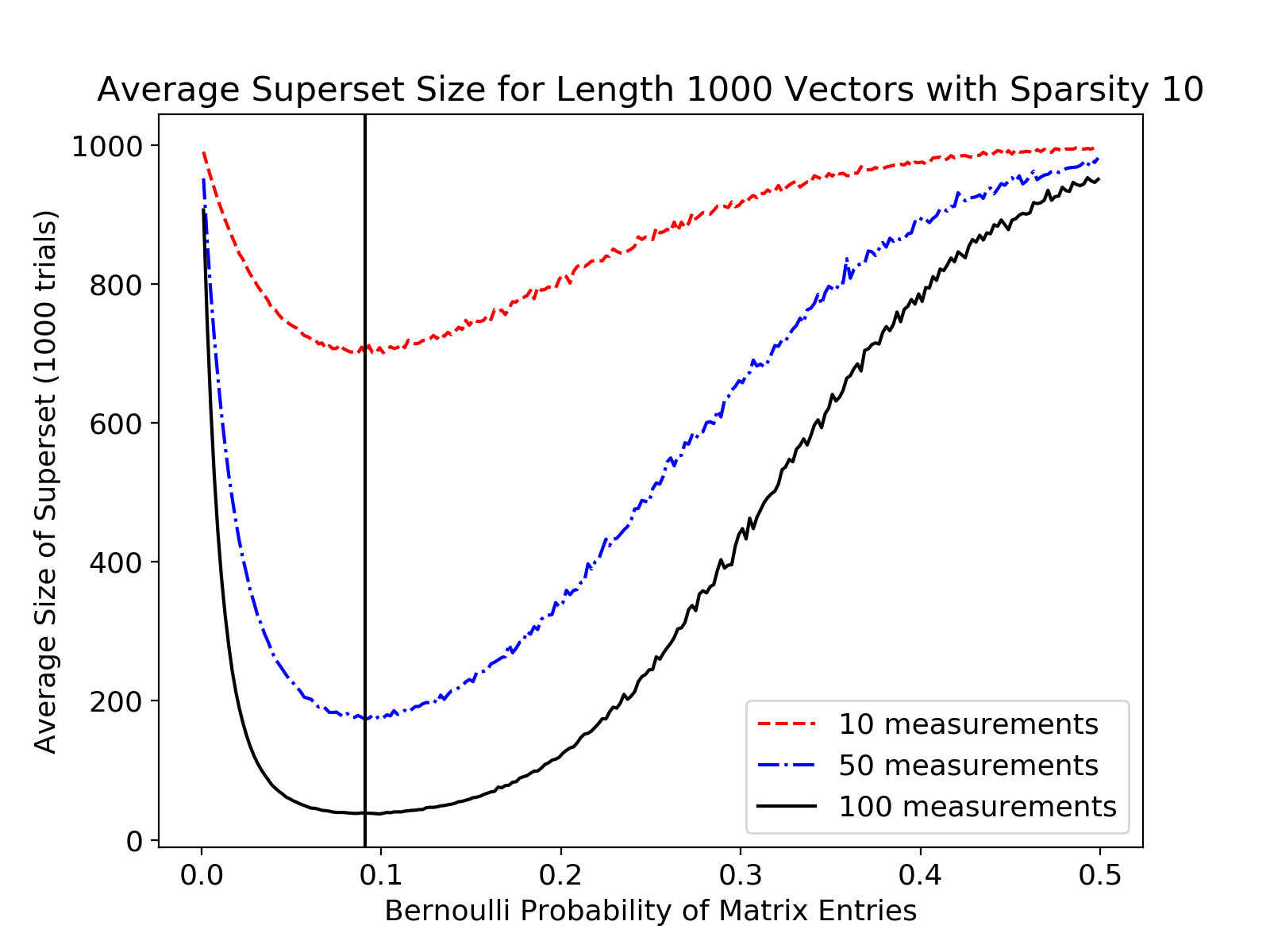}
                \caption{$n=1000, k=10$}
                \label{fig:bprob_k10}
        \end{subfigure}%
        \begin{subfigure}[H]{0.33\textwidth}
                \centering
                \includegraphics[width=.85\linewidth]{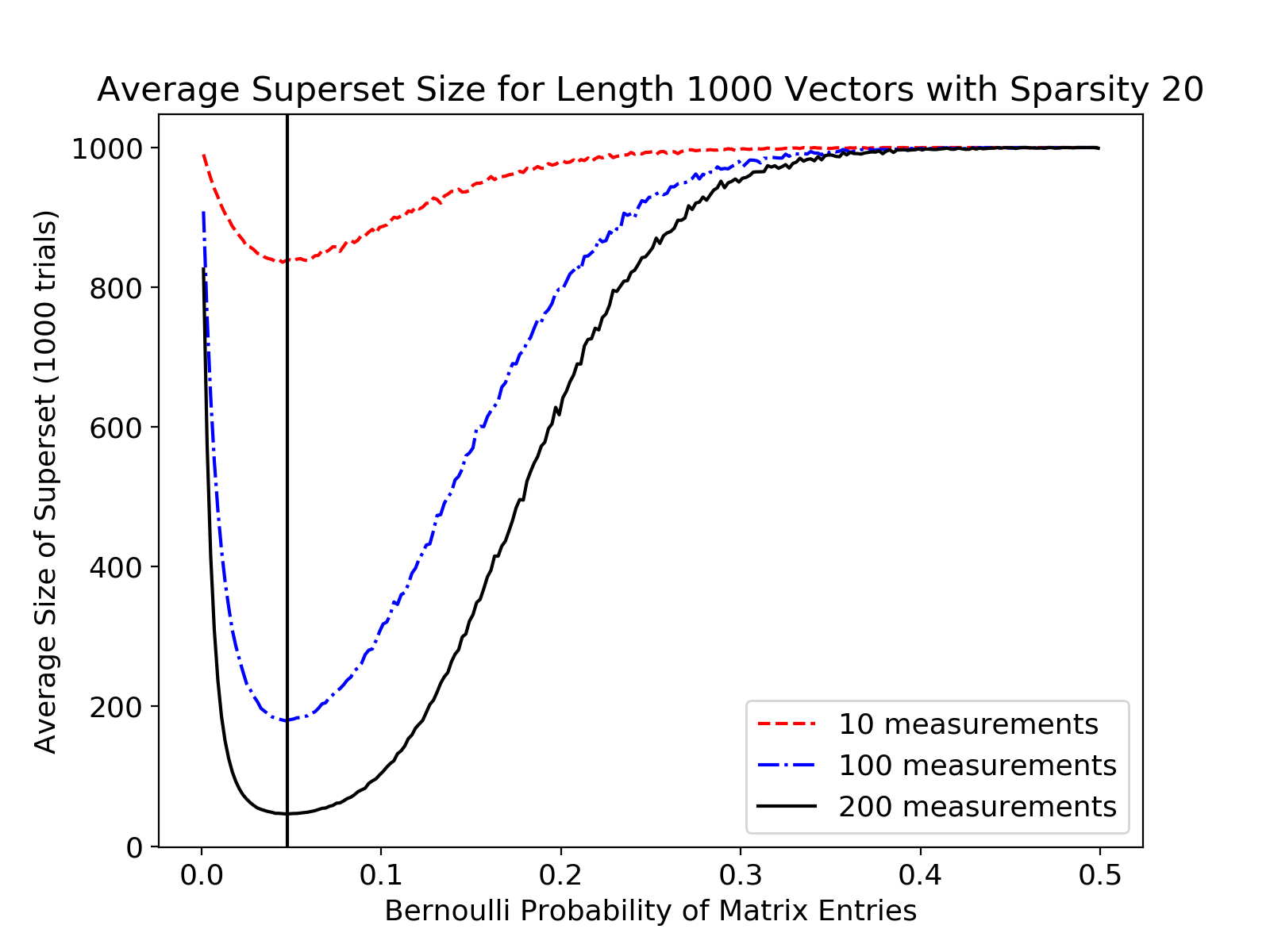}
                \caption{$n=1000, k=20$}
                \label{fig:bprob_k20}
        \end{subfigure}%
        \begin{subfigure}[H]{0.33\textwidth}
                \centering
                \includegraphics[width=.85\linewidth]{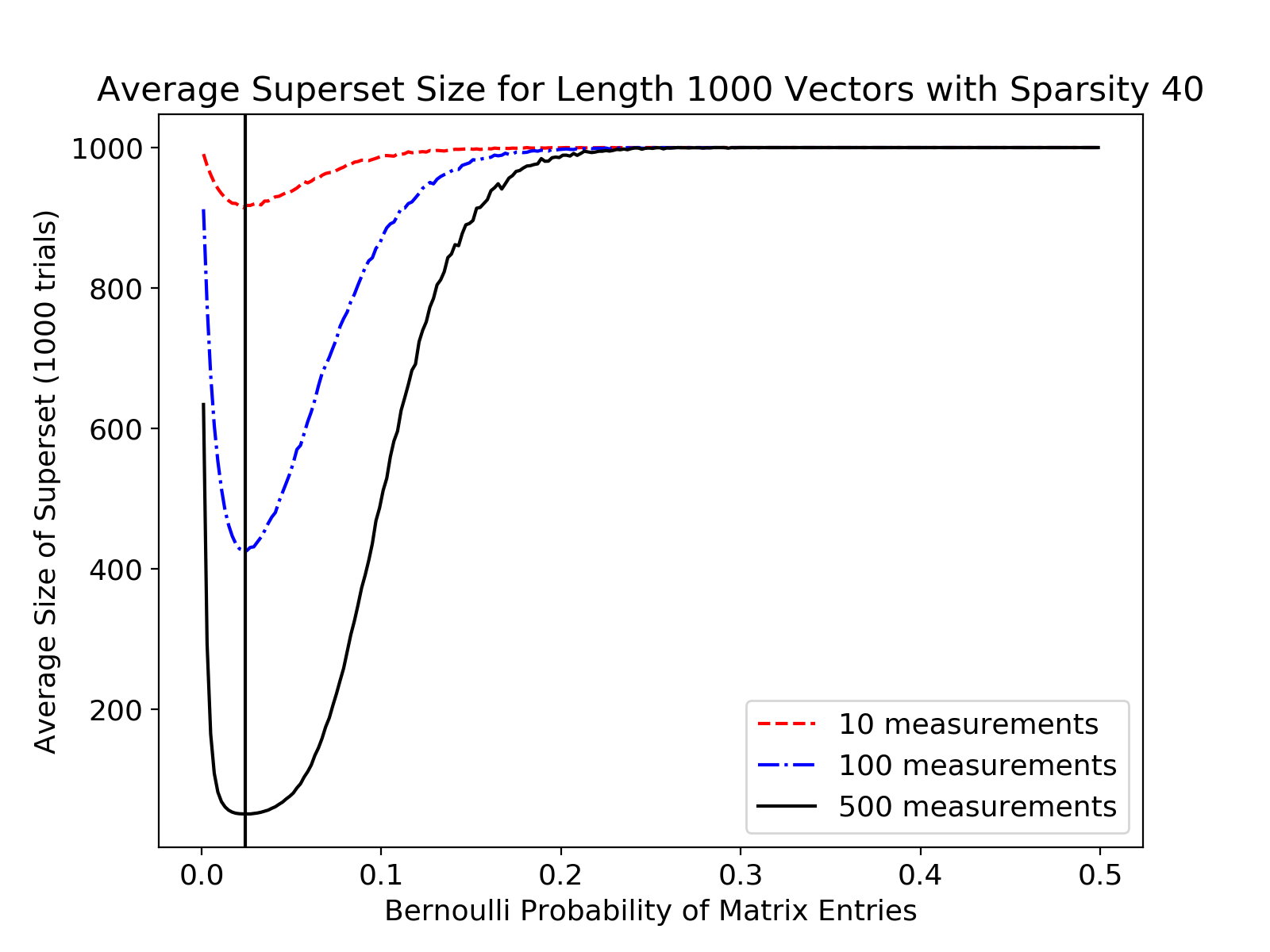}
                \caption{$n=1000, k=40$}
                \label{fig:bprob_k40}
        \end{subfigure}%
        \caption{Average size of superset following group testing decoding for different sparsity levels as Bernoulli probability of measurement matrix varies. Vertical line highlights $\frac{1}{k+1}$.}
        \label{fig:bern_probs}
\end{figure}

\Cref{fig:bern_probs} shows the average size of the superset using a matrix with Bernoulli entries (i.e. each value is 1 with probability $p$ and 0 otherwise) following a group testing decoding. The different lines represent different numbers of measurements used in the Bernoulli matrix, and different plots show different sparsity levels. All vectors had length 1000, and were constructed randomly by first choosing a size $k$ support set uniformly at random, then drawing a random value from $\mathcal{N}(0, 1)$ for each coordinate in the support set and normalizing so that $||\bfx||_2 = 1$. 1000 trials were performed for each tuple $(n, k, p)$ of values.

The vertical line overlaid atop the other curves in \cref{fig:bern_probs} indicates where the Bernoulli probability is equal to $\frac{1}{k+1}$. For all three sparsity levels, it appears that this value is very close to achieving the minimum size superset for a given number of measurements. Furthermore, the fact that the curves all have relatively wide basins around the minimum indicates that any value close to the minimum should perform fairly well.

\section{Sufficient Condition for Universal Support Recovery of Real Vectors}
\label{sec:suff_cond_reals}

The goal in this section is to give sufficient conditions on a measurement matrix in order to be able to recover a superset of the support of an unknown $k$-sparse signal $\bfx \in \R^n$ using 1-bit sign measurements, by generalizing the definition of ``Robust UFF'' given in \cite{ABK17}.

In this section we will work primarily with matrix columns rather than rows, so to this end for any matrix $B \in \R^{m \times n}$, here we let $B_j$ denote its $j$-th column. For any sets $X \subseteq [m]$ and $Y \subseteq [n]$, let $B[X:Y]$ denote the submatrix of $B$ restricted to rows indexed by $X$ and columns indexed by $Y$. Let $\wt(\bfx)$ denote the size of the support of $x$, i.e. $\wt(\bfx) = |\supp{\bfx}|$. We say $\bfx$ has full support if $\wt(\bfx) = n$. 

 In order to recover the superset of the support of $x$ using the sign measurements $\sign(B \bfx) \in \{-1, 0, 1\}^m$, we use the algorithm of \cite{ABK17} (Algorithm~\ref{alg:abk}). 
 For any subset of $k$ columns $S \subset [n]$, $|S| \le k$, define $T_S := \set{j \in [n]\setminus S : | \supp{B_j} \cap \left( \cup_{i \in S} \supp{B_i} \right) | \ge \frac 12 \wt(B_j)}$. These are the columns outside of the subset $S$ that have large intersection with the union of the $k$ columns indexed by $S$. 

\cite{ABK17} show that if $B$ is a Robust UFF with sufficient parameters, then their algorithm recovers the exact support of $x$.  Algorithm~\ref{alg:abk} computes the intersection of the support of each column $B_j$ with the output $b:=\sign(B \bfx)$. It includes the index $j$ in the estimated support if the intersection is sufficiently large. The property of a Robust UFF ensures that the estimated support is exactly the support of $\bfx$. 
\begin{algorithm}
\begin{algorithmic}
\STATE Input: $B:(n, m, d, k, 1/2)$-Robust UFF, $\bfx \in \R^n$, unknown $k$-sparse vector.
\STATE Let $\bfb:=\sign(Bx)$.
\STATE $\hat{S} = \emptyset$.
\STATE for $j \in [n]$,
\STATE \quad if $ |\supp{B_j} \cap \supp{\bfb} | > \frac d2$, 
\STATE \qquad $\hat{S} \leftarrow \{ j \}$
\STATE Return $\hat{S}$.
\end{algorithmic} 
\caption{Support recovery via Robust-UFF} 
\label{alg:abk}
\end{algorithm}
 
We relax the definition of an $(n, m, d, k, \alpha)$-Robust UFF to allow a few false positives, since we only require a superset of the support of $x$ rather than the exact support.  The allowable size of $T_S$ controls the number of false positives. Note that allowing $| T_S | \geq 1$ might induce some false negatives as well, thus to avoid this possibility we need to ensure that no column of $B$ in the support of $\bfx$ has too many zero test results. In general, zero test results can occur when $\bfx$ lies in the nullspace of many rows of $B$ that have a nonempty intersection with the support of $\bfx$. We construct the matrix $B$ to avoid such situations. 

For any subset $S \subseteq [n]$, and any $j \in T_S$, define $L_{S, j} := \{ t \in \supp{B_j} \cap (\cup_{i \in S} \supp{B_i}  )\} \subseteq [m]$. These are the rows in the support of $B_j$ that intersect with the support of the columns of $B$ indexed by $S$. In order to ensure that the algorithm does not introduce any false negatives, we want the output vector $\bfb$ to have not many zeros in rows corresponding to $L_{S, j}$.  
Let us define $A_{S, j} := B[L_{S, j}: S \cup \{j\}]$ to be the matrix restricted to the rows in $L_{S, j}$ and columns of $S \cup \{j\}$. Note that since $j \in T_S$, $|L_{S, j} | \ge \frac {\wt(B_j)}{2} $, therefore $A_{S,j}$ has at least $\frac {\wt(B_j)}{2}$ rows.
We now define a list-Robust UFF as follows: 

\begin{definition}[List-RUFF]
\label{def:list_ruff}
A real matrix $B \in \R^{m\times n}$ is called an $(m, n, d, k, 1/2, \ell)$-list Robust UFF if $\wt(B_j) = d$ for all $j \in [n]$, and for all subsets $S \subseteq [n]$, $ |S| \le k$, the following properties hold:
\begin{enumerate}
\item  $|T_S| < \ell$. 
\item  For any $j \in T_S$, and any $\bfx \in \R^{|S|}$ with full support,  
$\wt(A_{S, j} \bfx) > |L_{S, j}| - \frac12 \wt(B_j) $.
\end{enumerate}
\end{definition}

The first condition ensures that the Algorithm~\ref{alg:abk} introduces at most $\ell$ false positives. The second condition is used to ensure that no $k$-sparse vector $\bfx$ is in the nullspace of too many rows of $B$, and therefore Algorithm~\ref{alg:abk} will not yield any false negatives. 

Next we show that Algorithm~\ref{alg:abk} recovers a superset of size at most $k+\ell$ given a measurement matrix $B$ which is an $(m, n, d, k, 1/2, \ell)$-list RUFF.

\begin{theorem}
Let $\bfx \in \R^n$ be an unknown $k$-sparse vector with $\supp{\bfx} = S^*$. If $B$ is an $(n, m, d, k, 1/2, \ell)$-list RUFF, then Algorithm~\ref{alg:abk} returns $\hat{S}$ such that $S^* \subseteq \hat{S} \subseteq S^* \cup T_{S^*}$.
\end{theorem}

\begin{proof}
We first show that $\hat{S} \subseteq S^* \cup T_{S^*}$. We in fact prove the contrapositive, i.e. if $j \notin S^* \cup T_{S^*}$, then $j \notin \hat{S}$. Let $j \in [n] \setminus (S^* \cup T_{S^*})$. By definition of $T_{S^*}$, we know that $\supp{B_j}$ does not intersect $\cup_{i \in S^*} \supp{B_i}$ in too many places, i.e. 
$\lvert \supp{B_j} \cap \left( \cup_{i \in S^*} \supp{B_i} \right) \rvert < \frac{\wt(B_j)}{2}$. 
Consider all the rows $ t \in \supp{B_j} \setminus \left( \cup_{i \in S^*} \supp{B_i} \right)$. Note that for all these rows, $b_t = 0$. Therefore, 
\begin{align*}
\lvert \supp{b} \cap \supp{B_j} \rvert &\le \lvert \supp{B_j} \rvert - \lvert  \supp{B_j} \setminus \left( \cup_{i \in S^*} \supp{B_i} \right) \rvert \\
&= \lvert \supp{B_j} \cap \left( \cup_{i \in S^*} \supp{B_i} \right) \rvert <  \frac{\wt(B_j)}{2}. 
\end{align*}
From Algorithm~\ref{alg:abk}, it then follows that $j \notin \hat{S}$.

To show that every $j \in S^*$ is included in $\hat{S}$, we need to show that for every such $j$, $\lvert \supp{b} \cap \supp{B_j} \rvert >  \frac{ \wt(B_j)}{2}$. This is equivalent to showing that there are not too many zeros in the rows of $b$ corresponding to rows in $\cup_{i \in S^*} \supp{B_i}$. Let $j \in S^*$ be any column in the support of $\bfx$. Let us partition $\supp{B_j}$ into two groups. 
Let $S^*_{j} := S^* \setminus \{j\}$. 
Define 
\begin{align*}
G_1:= &\{ t \in \supp{B_j} \cap \left( \cup_{i \in S^*_{j}}  \supp{B_i} \right) \}, \mbox{ and } \\
G_2: = &\supp{B_j} \setminus G_1 = \{ t \in \supp{B_j} \setminus \left( \cup_{i \in S^*_{j}}  \supp{B_i} \right) \}.
\end{align*}
Note that for all $t \in G_2$, $b_t \neq 0$ since $b_t = \bfx_j \cdot B_j(t) \neq 0$ since $j \in \supp{\bfx}$. Therefore, $G_2 \subseteq \supp{\bfb} \cap \supp{B_j}$. We can without loss of generality assume that $j \in T_{S_j^*}$. 
Otherwise, by definition of $T_{S_j^*}$ it follows that 
$\lvert G_2 \rvert > \frac{\wt(B_j)}{2}$, and Algorithm~\ref{alg:abk} includes $j \in \hat{S}$.

We now show that $\bfb_t \neq 0$ for many $t \in G_1$. In particular, we show that $\bfb_t$ is zero for at most $\frac{\wt(B_j)}{2}$ indices in $G_1$. This follows from the property of the list-RUFF. 
Consider the following submatrix of $B$,
$A_{S_j^*, j} := B[G_1,S^*] = B \left[ L_{ S_j^*, j} : S_j^* \cup \{j\} \right]$.
Since $j \in T_{S_j^*}$, $|G_1| > \wt(B_j) / 2$, and therefore $A_{S_j^*, j}$ has at least $\wt(B_j) / 2$ rows, and at most $k$ columns.  

From the definition of list-RUFF, we know that for any $\bfz \in \R^{|S^*|}$ with full support, $\wt(A_{S_j^*, j}~\bfz) > |L_{S, j}| - \frac12 \wt(B_j) = |G_1| - \frac12 \wt(B_j)$. Therefore, for $\bfx$ that is supported on $S^*$, $\bfb_t \neq 0$ for at least  $|G_1| - \frac12 \wt(B_j)$ indices in $G_1$. 

Combining these observations, it follows that 
\begin{align*}
\lvert \supp{\bfb} \cap \supp{B_j} \rvert > |G_1| - \frac12 \wt(B_j) + |G_2| = \frac12 \wt(B_j).
\end{align*}
Therefore the fact that $j \in \hat{S}$ follows from Algorithm~\ref{alg:abk}.

 \end{proof}
 
 In light of this, a possible direction for improving the current upper bound for universal approximate recovery of real vectors would be to show the existence of $(m, n, d, k, 1/2, \calO{k})$-list RUFFs with $m = o(k^2 \log(\frac{n}{k}))$. This would immediately yield a measurement matrix with $\calO{m + \frac{k}{\epsilon} \log \frac{k}{\epsilon}}$ rows that could be used for universal $\epsilon$-approximate recovery. We show below via a simple probabilistic construction that matrices satisfying the first property in \cref{def:list_ruff} with $m = O(k \log n)$ and $\ell = O(k)$ exist, but leave open the question of whether $O(k \log n)$ rows suffices also for the second property, or whether $O(k^2 \log n)$ rows are necessary.
 
 \begin{theorem}
 \label{thm:list_ruff_prop1}
 There exist matrices $B \in \R^{m \times n}$ satisfying $\wt(B_j) = \frac{m}{k}$ for all columns $B_j$ and for every subset of columns $S \subseteq [n]$, $|S| \leq k$, we have $|T_S| < \ell$, under the assumptions that $m = \Omega(k \log n)$, $k = o(n / \log(n))$, and $\ell = \Omega(k)$.
 \end{theorem}
 \begin{proof}
 We will construct $B$ by drawing a set $S_j \subseteq [m]$ of size $d = \frac{m}{k}$ uniformly at random among all such sets for each column of $B$. If $i \in S_j$ then we set the $i$th entry of $B_j$ to 1, otherwise 0. Now we must show that with probability less than 1 there does not exist any subset $S$ of at most $k$ columns of $B$ with $|T_S| \geq \ell = \Omega(k)$.
 
 Recall that by definition,
 \begin{equation*}
 T_S = \set{j \in [n] \setminus S : |\supp{B_j} \cap (\cup_{i \in S} \supp{B_i})| \geq \frac{1}{2} \wt(B_j)},
 \end{equation*}
 or in other words, $T_S$ is the set of ``confusable'' columns for the subset $S$ of columns of $B$. The event that we wish to avoid is that there exists a set $S$ of $k + \ell$ ``bad'' columns for which the union of the supports of a subset $S' \subseteq S$ of $k$ of those columns has a large intersection with the supports of all of the remaining $\ell$ columns. Since the columns of $B$ are all chosen independently, we have
 \begin{eqnarray}
 && \Pr[B \textrm{ has a bad set $S$ of $k + \ell$ columns}] \\
 &\leq& \binom{n}{\ell + k} \Pr[S \subseteq \Col(B) \textrm{ is a bad set of $k+\ell$ columns}] \\
 &\leq& \binom{n}{\ell + k} \binom{\ell + k}{k} \Pr[\textrm{ for all $\ell$ columns $B_i$ in $S \setminus S'$, $i \in T_{S'}$}] \\
 &\leq& \binom{n}{\ell + k} \binom{\ell + k}{k} (\Pr[i \in T_{S'}])^{\ell}.
 \end{eqnarray}
 
Now we can assume we have a fixed set $S'$ of $k$ columns and another fixed column $B_i$, and we want to upper bound the probability that more than half the $d = \frac{m}{k}$ nonzero entries of $B_i$ lie in $\cup_{j \in S'} \supp{B_j}$. Let $X_j$ be the binary random variable that is equal to 1 if and only if the $j$th entry of $B_i$ is nonzero and lies in $\cup_{j \in S'} \supp{B_j}$. Since every column has weight exactly $d$, $|\cup_{j \in S'} \supp{B_j}| \leq kd$, thus for any $j$ $\Pr[X_j = 1] \leq \frac{kd}{n}.$ Then by linearity of expectation we conclude that
 \begin{equation}
 E[\sum_{j=1}^m X_j] = d \Pr[X_j = 1] \leq \frac{kd^2}{n}.
 \end{equation}
 
 While the $X_j$ are not independent, if some $X_j = 1$ then it is less likely that a different random variable $X_{j'} = 1$ as there are less coordinates remaining in $\cup_{j \in S'} \supp{B_j}$. Since the $X_j$ are negatively correlated we can apply a Chernoff bound:
 \begin{equation}
 \Pr[\sum_{j=1}^m X_j \geq \frac{n}{2m} E[\sum_{j=1}^m X_j]] < \left(\frac{e^{(n/2m) - 1}}{(n / 2m)^{n / 2m}}\right)^{m^2 / nk} < \left(\frac{2em}{n}\right)^{m / 2k}.
 \end{equation}
Note that
 \begin{equation}
 \frac{n}{2m} E[\sum_{j=1}^m X_j] \leq \frac{n}{2m} \cdot \frac{kd^2}{n} = \frac{d}{2},
 \end{equation}
so in order for the sum of the $X_j$ to exceed $\frac{d}{2}$ (which would mean the corresponding fixed column has large overlap with the union of the set of $k$ columns), it must also exceed $\frac{n}{2m} E[\sum_{j=1}^m X_j]$.
 
 Combining everything above,
 \begin{eqnarray}
 && \Pr[B \textrm{ has a bad set $S$ of $k + \ell$ columns}] \\
&\leq& \binom{n}{\ell + k} \binom{\ell + k}{k} (\Pr[B_i \in T_{S'}])^{\ell} \\
&\leq& \binom{n}{\ell + k} \binom{\ell + k}{k} \left(\frac{2em}{n}\right)^{(\ell m) / (2k)} \\
&\leq& \left(\frac{ne}{k+\ell}\right)^{k+\ell} \left(\frac{(\ell + k)e}{k}\right)^k \left(\frac{2em}{n}\right)^{(\ell m) / (2k)} \\
&\leq& \left(\frac{ne}{k}\right)^{2k + \ell} \left(\frac{2em}{n}\right)^{(\ell m) / (2k)},
\end{eqnarray}
and we can make this final quantity less than 1 by choosing $m = ck \log n$ for an appropriately large constant $c$, using our assumptions that $\ell = \Omega(k)$ and $k = o(n / (\log n))$.
\end{proof}

%% file: superset.bbl
\begin{thebibliography}{10}

\bibitem{ABK17}
Jayadev Acharya, Arnab Bhattacharyya, and Pritish Kamath.
\newblock Improved bounds for universal one-bit compressive sensing.
\newblock In {\em 2017 IEEE International Symposium on Information Theory
  (ISIT)}, pages 2353--2357. IEEE, 2017.

\bibitem{DBLP:journals/tit/AldridgeBJ14}
Matthew Aldridge, Leonardo Baldassini, and Oliver Johnson.
\newblock Group testing algorithms: Bounds and simulations.
\newblock {\em {IEEE} Trans. Information Theory}, 60(6):3671--3687, 2014.

\bibitem{AtiaS12}
George~K. Atia and Venkatesh Saligrama.
\newblock Boolean compressed sensing and noisy group testing.
\newblock {\em {IEEE} Trans. Information Theory}, 58(3):1880--1901, 2012.

\bibitem{DBLP:conf/ciss/BoufounosB08}
Petros Boufounos and Richard~G. Baraniuk.
\newblock 1-bit compressive sensing.
\newblock In {\em 42nd Annual Conference on Information Sciences and Systems,
  {CISS} 2008, Princeton, NJ, USA, 19-21 March 2008}, pages 16--21. {IEEE},
  2008.

\bibitem{Che13}
Mahdi Cheraghchi.
\newblock Noise-resilient group testing: Limitations and constructions.
\newblock {\em Discrete Applied Mathematics}, 161(1-2):81--95, 2013.

\bibitem{de2005optimal}
Annalisa De~Bonis, Leszek Gasieniec, and Ugo Vaccaro.
\newblock Optimal two-stage algorithms for group testing problems.
\newblock {\em SIAM Journal on Computing}, 34(5):1253--1270, 2005.

\bibitem{DBLP:journals/tit/Donoho06}
David~L. Donoho.
\newblock Compressed sensing.
\newblock {\em {IEEE} Trans. Information Theory}, 52(4):1289--1306, 2006.

\bibitem{du2000combinatorial}
D.~Du and F.~Hwang.
\newblock {\em Combinatorial Group Testing and Its Applications}.
\newblock Applied Mathematics. World Scientific, 2000.

\bibitem{GNJN13}
Sivakant Gopi, Praneeth Netrapalli, Prateek Jain, and Aditya Nori.
\newblock One-bit compressed sensing: Provable support and vector recovery.
\newblock In {\em International Conference on Machine Learning}, pages
  154--162, 2013.

\bibitem{DBLP:conf/ciss/HauptB11}
Jarvis~D. Haupt and Richard~G. Baraniuk.
\newblock Robust support recovery using sparse compressive sensing matrices.
\newblock In {\em 45st Annual Conference on Information Sciences and Systems,
  {CISS} 2011, The John Hopkins University, Baltimore, MD, USA, 23-25 March
  2011}, pages 1--6. {IEEE}, 2011.

\bibitem{JLBB13}
Laurent Jacques, Jason~N Laska, Petros~T Boufounos, and Richard~G Baraniuk.
\newblock Robust 1-bit compressive sensing via binary stable embeddings of
  sparse vectors.
\newblock {\em IEEE Transactions on Information Theory}, 59(4):2082--2102,
  2013.

\bibitem{landau1967sampling}
HJ~Landau.
\newblock Sampling, data transmission, and the nyquist rate.
\newblock {\em Proceedings of the IEEE}, 55(10):1701--1706, 1967.

\bibitem{DBLP:conf/aistats/Li16}
Ping Li.
\newblock One scan 1-bit compressed sensing.
\newblock In Arthur Gretton and Christian~C. Robert, editors, {\em Proceedings
  of the 19th International Conference on Artificial Intelligence and
  Statistics, {AISTATS} 2016, Cadiz, Spain, May 9-11, 2016}, volume~51 of {\em
  {JMLR} Workshop and Conference Proceedings}, pages 1515--1523. JMLR.org,
  2016.

\bibitem{Maz16}
Arya Mazumdar.
\newblock Nonadaptive group testing with random set of defectives.
\newblock {\em {IEEE} Trans. Information Theory}, 62(12):7522--7531, 2016.

\bibitem{DBLP:journals/tit/PlanV13}
Yaniv Plan and Roman Vershynin.
\newblock Robust 1-bit compressed sensing and sparse logistic regression: {A}
  convex programming approach.
\newblock {\em {IEEE} Trans. Information Theory}, 59(1):482--494, 2013.

\bibitem{PR11}
Ely Porat and Amir Rothschild.
\newblock Explicit nonadaptive combinatorial group testing schemes.
\newblock {\em {IEEE} Trans. Information Theory}, 57(12):7982--7989, 2011.

\bibitem{DBLP:conf/ita/ShiCGTN16}
Hao{-}Jun~Michael Shi, Mindy Case, Xiaoyi Gu, Shenyinying Tu, and Deanna
  Needell.
\newblock Methods for quantized compressed sensing.
\newblock In {\em 2016 Information Theory and Applications Workshop, {ITA}
  2016, La Jolla, CA, USA, January 31 - February 5, 2016}, pages 1--9. {IEEE},
  2016.

\bibitem{DBLP:conf/nips/SlawskiL15}
Martin Slawski and Ping Li.
\newblock b-bit marginal regression.
\newblock In Corinna Cortes, Neil~D. Lawrence, Daniel~D. Lee, Masashi Sugiyama,
  and Roman Garnett, editors, {\em Advances in Neural Information Processing
  Systems 28: Annual Conference on Neural Information Processing Systems 2015,
  December 7-12, 2015, Montreal, Quebec, Canada}, pages 2062--2070, 2015.

\bibitem{tibshirani1996regression}
Robert Tibshirani.
\newblock Regression shrinkage and selection via the lasso.
\newblock {\em Journal of the Royal Statistical Society: Series B
  (Methodological)}, 58(1):267--288, 1996.

\end{thebibliography}
